\documentclass[12pt]{article}
\usepackage{a4,amsmath,amssymb,amsthm,graphicx,pdflscape,caption}
\usepackage[table]{xcolor}

\usepackage{ytableau,url}
\usepackage{tikz}
\usetikzlibrary{intersections,calc,matrix,arrows, decorations.markings}

\newtheorem{theorem}{Theorem}[section]

\newtheorem{prop}[theorem]{Proposition}
\newtheorem{conj}[theorem]{Conjecture}

\begin{document}
\begin{titlepage}
 \hfill {\tt arXiv:1607.nnnn} \\[4pt]
  \mbox{}\hfill \hfill{\fbox{\textbf{v1; July 2016 }}}
\begin{center}
\textsf{\large On a square-ice analogue of plane partitions}
\end{center}
\bigskip
\begin{center}
\textsf{Suresh Govindarajan$^\dagger$, Anthony J. Guttmann$^*$ and Varsha Subramanyan$^\dagger$}\\[5pt]
$^\dagger$Department of Physics, \\ Indian Institute of Technology Madras,\\ Chennai 600036, India \\
Email: suresh@physics.iitm.ac.in  \\[3pt]
and \\[3pt]
$^*$ARC Centre of Excellence for Mathematics and Statistics of Complex Systems,\\ School of Mathematics and Statistics,\\
The University of Melbourne, \\
 Victoria 3010, Australia \\
 Email: tonyg@ms.unimelb.edu.au
\end{center}
\begin{abstract}
We study a one-parameter family  ($\ell=1,2,3,\ldots$) of configurations  that are square-ice analogues of plane partitions. Using an algorithm due to Bratley and McKay, we carry out exact enumerations in order to study their asymptotic behaviour and establish, via Monte Carlo simulations as well as explicit bounds, that the asymptotic behaviour is similar to that of plane partitions. We finally carry out a series analysis and provide independent estimates for the asymptotic behaviour.
\end{abstract}
\end{titlepage}

\section{Introduction}

A seller of oranges arranges his oranges in the following fashion. The top layer has a row of $\ell$ ($=1,2,3,\ldots$)  oranges, the second layer has oranges forming a  $2\times (\ell+1)$ rectangle and in the $k$-th layer, the oranges form a $k\times (k+\ell-1)$ rectangle (see Figure \ref{orangestack}). We call the parameter $\ell$ the \textit{width} of a configuration. Assuming that there are infinitely many layers, in how many ways can one remove $n$ oranges without upsetting any other oranges? Denote this by number by $a_\ell(n)$. We study properties of the sequences $a_\ell(n)$ in the paper.

\begin{figure}[htb]
\begin{center}
\includegraphics[height=3in]{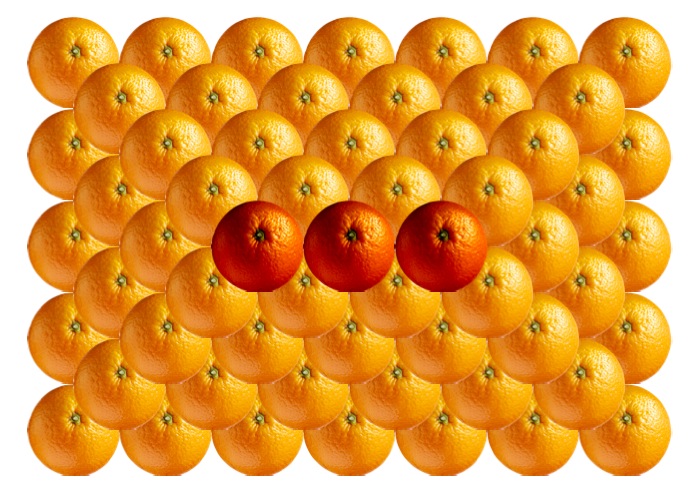}
\end{center}
\caption{A stack with five layers of oranges  and width $\ell=3$.} \label{orangestack}
\end{figure}

In an alternative definition of the same problem in terms of height functions (as given  in section 2), one observes that the local conditions on the height function are the same as those for plane and pyramid partitions. Propp in a post in the domino forum  \cite{Propp2014} in August 2014 asked whether one can find explicit formulae for the generating functions as is known in the case of plane and pyramid partitions  \cite{Macmahon1896memoir,Young2007,Szendroi2008}.  The reformulation in terms of stacking oranges is due to R. Kenyon and the variant involving the number of oranges is due to Young \cite{Propp2014}.

In this paper, we address this issue by explicitly generating numbers for width $a_\ell(n)$ for $\ell=1,2,\ldots,6$ by adapting an algorithm due to Bratley-McKay  \cite{Bratley:1967a}. We have been unable to find  an explicit formula for the generating function. In the absence of a  formula for the generating function, we address the following two questions  in this paper.
\begin{enumerate}
\item For fixed $n$, what are the properties of $a_\ell(n)$? 
\item For fixed width $\ell$, what is the asymptotic behaviour of $a_\ell(n)$?
\end{enumerate}

The organisation of the the paper is as follows. After the introductory section where we state the problem at hand, in section 2, we give a formal definition of the problem and study the properties of $a_\ell(n)$ for fixed $n$. We obtain an interesting conjecture for $\ell\geq \lceil{n/2}\rceil$. In section 3,
we first set upper and lower bounds on $a_\ell(n)$ and numerically estimate the asymptotic behaviour using  transition matrix Monte Carlo simulations for $\ell \in[1,6]$. In section 4, we analyse the series of numbers obtained from exact enumeration to independently estimate the asymptotic behaviour as well as extrapolate the sequence of coefficients in order to obtain the next ten coefficients for $a_1(n)$. We conclude with a few remarks in section 5. Appendix A tabulates the results of our exact enumerations. In appendix B, we introduce a sub-class of plane partitions that appears naturally in this work and set bounds on the asymptotic behaviour of these restricted plane partitions.

\section{Definitions and exact results}

\textbf{Definition:} Let $v=(x,y)\in \mathbb{Z}^2$ and for fixed $\ell=1,2,3,\ldots$, following  \cite{Propp2014} define 
$$
h^{(\ell)}_0(v)=\begin{cases}
|x|+|y| & x<0 \\
|x+y| & 0\leq x < \ell\\
|y+\ell -1| + |x-\ell+1| & x\geq \ell
\end{cases}\ .
$$
The \textit{height} function $h$  on $\mathbb{Z}^2$ is an integer-valued function that agrees with $h^{(\ell)}_0$ almost everywhere (i.e., at all but finitely many places), is greater than or equal to $h^{(\ell)}_0$ everywhere, and satisfies the
condition that if $u$ and $v$ are adjacent locations in $\mathbb{Z}^2$,
$|h(u)-h(v)|=1$. The last condition is called the \textit{ice rule}.

\noindent \textbf{Definition:} Define the volume of the height function as follows:
\begin{equation}
n:=\sum_{(x,y)\in \mathbb{Z}^2} \frac12 \left(h(x,y)-h^{(\ell)}_0(x,y)\right)\ .
\end{equation}
\noindent \textbf{Definition:} Let $a_\ell(n)$ denote the number of height functions with volume $n$ for an initial configuration of width $\ell$.

 \begin{figure}[hbt]
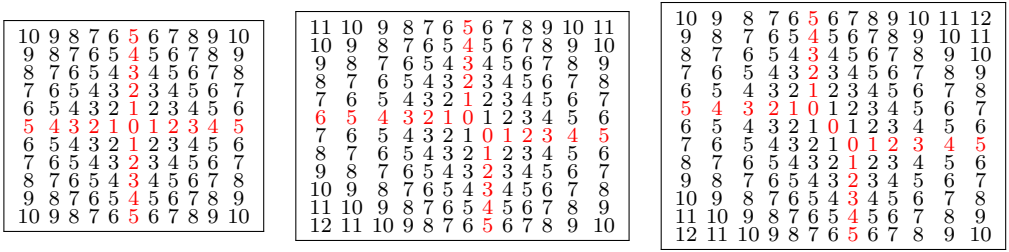

 $$
 \boxed{
 \begin{smallmatrix}
10 & 9 & 8 & 7 & 6 & \color{red}5 & 6 & 7 & 8 & 9 & 10 \\
9 & 8 & 7 & 6 & 5 & \color{red}4 & 5 & 6 & 7 & 8 & 9 \\
8 & 7 & 6 & 5 & 4 & \color{red}3 & 4 & 5 & 6 & 7 & 8 \\
7 & 6 & 5 & 4 & 3 &\color{red} 2 & 3 & 4 & 5 & 6 & 7 \\
6 & 5 & 4 & 3 & 2 & \color{red} 1 & 2 & 3 & 4 & 5 & 6 \\
\color{red}5 & \color{red}4 &\color{red} 3 &\color{red} 2 & \color{red}1 & {\color{red} 0} &\color{red} 1 & \color{red}2 & \color{red}3 & \color{red}4 & \color{red}5 \\
6 & 5 & 4 & 3 & 2 & \color{red}1 & 2 & 3 & 4 & 5 & 6 \\
7 & 6 & 5 & 4 & 3 & \color{red}2 & 3 & 4 & 5 & 6 & 7 \\
8 & 7 & 6 & 5 & 4 & \color{red}3 & 4 & 5 & 6 & 7 & 8 \\
9 & 8 & 7 & 6 & 5 & \color{red}4 & 5 & 6 & 7 & 8 & 9 \\
10 & 9 & 8 & 7 & 6 & \color{red}5 & 6 & 7 & 8 & 9 & 10
\end{smallmatrix}
}\quad
\boxed{
\begin{smallmatrix}
 11 & 10 & 9 & 8 & 7 & 6 & \color{red}5 & 6 & 7 & 8 & 9 & 10 & 11 \\
 10 & 9 & 8 & 7 & 6 & 5 & \color{red}4 & 5 & 6 & 7 & 8 & 9 & 10 \\
 9 & 8 & 7 & 6 & 5 & 4 & \color{red}3 & 4 & 5 & 6 & 7 & 8 & 9 \\
 8 & 7 & 6 & 5 & 4 & 3 & \color{red}2 & 3 & 4 & 5 & 6 & 7 & 8 \\
 7 & 6 & 5 & 4 & 3 & 2 &\color{red} 1 & 2 & 3 & 4 & 5 & 6 & 7 \\
\color{red} 6 &\color{red} 5 &\color{red} 4 &\color{red} 3 & \color{red}2 & \color{red}1 & {\color{red} 0} & 1 & 2 & 3 & 4 & 5 & 6 \\
 7 & 6 & 5 & 4 & 3 & 2 & 1 & {\color{red} 0} & \color{red}1 & \color{red}2 &\color{red} 3 &\color{red} 4 &\color{red} 5 \\
 8 & 7 & 6 & 5 & 4 & 3 & 2 & \color{red}1 & 2 & 3 & 4 & 5 & 6 \\
 9 & 8 & 7 & 6 & 5 & 4 & 3 & \color{red}2 & 3 & 4 & 5 & 6 & 7 \\
 10 & 9 & 8 & 7 & 6 & 5 & 4 & \color{red}3 & 4 & 5 & 6 & 7 & 8 \\
 11 & 10 & 9 & 8 & 7 & 6 & 5 & \color{red}4 & 5 & 6 & 7 & 8 & 9 \\
 12 & 11 & 10 & 9 & 8 & 7 & 6 & \color{red}5 & 6 & 7 & 8 & 9 & 10
\end{smallmatrix}
}
\quad
\boxed{ \begin {smallmatrix} 
10 & 9 & 8 & 7 & 6 & \color{red}5 & 6 & 7 & 8 & 9 & 10 & 11 & 12 \\
9 & 8 & 7 & 6 & 5 & \color{red}4 & 5 & 6 & 7 & 8 & 9 & 10 & 11 \\
8 & 7 & 6 & 5 & 4 & \color{red}3 & 4 & 5 & 6 & 7 & 8 & 9 & 10 \\
7 & 6 & 5 & 4 & 3 & \color{red}2 & 3 & 4 & 5 & 6 & 7 & 8 & 9 \\
6 & 5 & 4 & 3 & 2 &\color{red} 1 & 2 & 3 & 4 & 5 & 6 & 7 & 8 \\
\color{red}5 & \color{red}4 & \color{red}3 & \color{red}2 & \color{red}1 &{\color{red} 0} & 1 & 2 & 3 & 4 & 5 & 6 & 7 \\
6 & 5 & 4 & 3 & 2 & 1 & {\color{red} 0} & 1 & 2 & 3 & 4 & 5 & 6 \\
7 & 6 & 5 & 4 & 3 & 2 & 1 & {\color{red} 0} & \color{red}1 & \color{red}2 & \color{red}3 & \color{red}4 &\color{red} 5 \\
8 & 7 & 6 & 5 & 4 & 3 & 2 & \color{red}1 & 2 & 3 & 4 & 5 & 6 \\
9 & 8 & 7 & 6 & 5 & 4 & 3 & \color{red}2 & 3 & 4 & 5 & 6 & 7 \\
10 & 9 & 8 & 7 & 6 & 5 & 4 & \color{red}3 & 4 & 5 & 6 & 7 & 8 \\
11 & 10 & 9 & 8 & 7 & 6 & 5 & \color{red}4 & 5 & 6 & 7 & 8 & 9 \\
12 & 11 & 10 & 9 & 8 & 7 & 6 & \color{red}5 & 6 & 7 & 8 & 9 & 10
\end {smallmatrix}
}
$$
\caption{Initial height functions $h_0^\ell$ for width $\ell=1,2,3$ inside a square. The red numbers partition the plane into four parts which we label as the NE, NW, SW and the SE parts. The creases are indicated in red.  } \label{refconfigs}
\end{figure}
\ytableausetup{centertableaux,smalltableaux}

\subsection{The reduced height function}

 \noindent \textbf{Definition:} Define the reduced height function (on $\mathbb{Z}^2$) as follows:
\begin{equation}
r(x,y)=\frac12 \left(h(x,y)-h^{(\ell)}_0(x,y)\right)\ ,
\end{equation}
where $r(x,y)$ is a non-negative integer. Call the set of points $(x,-x)$ (for $0\leq x < \ell$) where the topmost oranges lie, the \textit{central crease}. The \textit{northern crease} is the set of points $(0,y)$ with $y>0$ and  the \textit{western crease} is the set of points $(x,0)$ with $x<0$. The \textit{eastern crease} refers to the points $(x+\ell-1,\ell-1)$ for $x>0$ and the \textit{southern crease} to the set of points  $(\ell-1,1-\ell+y)$ for $y<0$. These points located on the creases are indicated in red numbers in the reference configurations shown in Figure \ref{refconfigs}.
\begin{prop} The reduced height function is a weakly decreasing function as one moves away from the creases. Further, for unit steps along the N/S/E/W directions, it can change by at most one. 
\end{prop}
\begin{proof}
Since the creases split configurations into four parts, we shall pick one part, say the NE part, and prove this property. In the NE part, going away from the crease corresponds to increasing the $x$ or $y$ coordinate by one. Consider a pair of neighbouring points, $u=(x,y)$ and $v=(x+1,y)$. Since $h^{(\ell)}_0(v)-h^{(\ell)}_0(u)=1$, one has
\[
r(u)-r(v) = \tfrac12\big(h(u)-h(v) -h^{(\ell)}_0(u)+h^{(\ell)}_0(v)\big) = \tfrac12 \big(h(u)-h(v)+1\big)\ .
\]
Since $|h(v)-h(u)|=1$, we see that $(r(v)-r(u))$ is either $0$ or $-1$. A similar proof shows that this is true for all other cases as well.
\end{proof}
Thus, given a configuration with volume $n$, it can be broken up into  2 plane partitions and 2 skew plane partitions  with volumes $(n_1,n_2,n_3,n_4)$ where $\sum_{j=1}^4 n_j=n$. These plane partitions are not the most general ones as the height condition is stronger than the weakly decreasing  condition imposed for plane partitions (see Appendix B). We illustrate this split  in  Figure \ref{splitconfig} for a random configuration with $\ell=6$ and volume$=120$. 
\begin{figure}[htb]
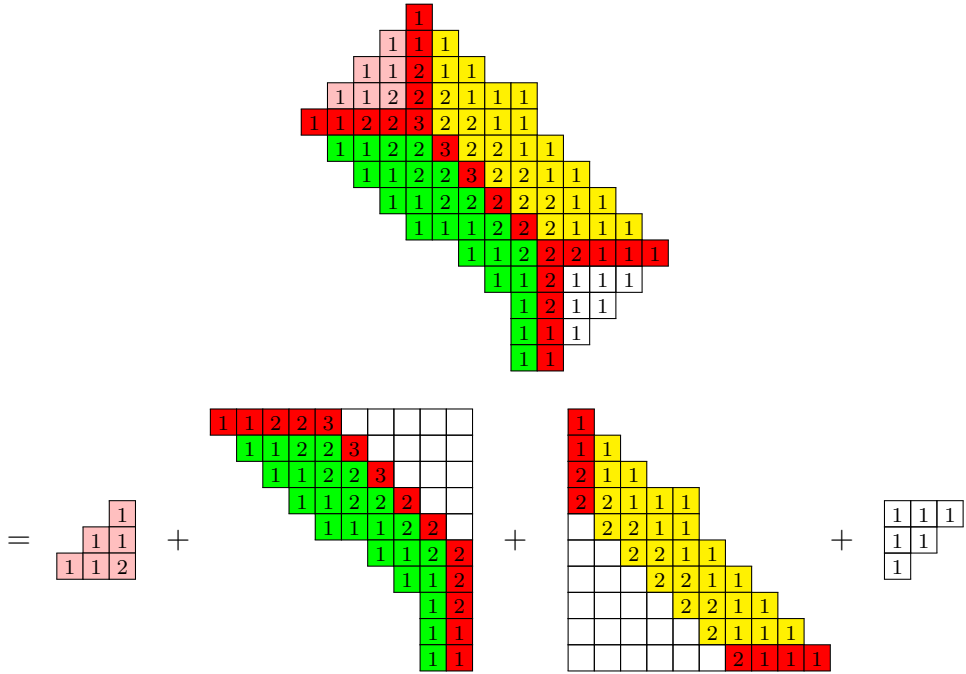

\begin{center}
\ytableaushort{
\none\none\none\none{*(red) 1},
\none\none\none{*(pink) 1} {*(red) 1} {*(yellow) 1},
\none\none{*(pink) 1} {*(pink) 1} {*(red) 2} {*(yellow) 1} {*(yellow) 1},
\none{*(pink) 1}{*(pink) 1}{*(pink) 2}{*(red) 2}{*(yellow) 2}{*(yellow) 1}{*(yellow) 1}{*(yellow) 1},
{*(red) 1}{*(red) 1}{*(red) 2}{*(red) 2}{*(red) 3}{*(yellow) 2}{*(yellow) 2}{*(yellow) 1}{*(yellow) 1},
\none {*(green) 1}{*(green) 1}{*(green) 2}{*(green)  2}{*(red) 3}{*(yellow)2}{*(yellow)2}{ *(yellow)1}{*(yellow) 1},
\none\none{*(green) 1}{*(green) 1}{*(green) 2}{*(green) 2}{*(red) 3}{*(yellow) 2}{*(yellow) 2}{*(yellow) 1}{*(yellow) 1},
\none\none\none{*(green)1}{*(green) 1}{*(green) 2}{*(green) 2}{*(red) 2}{*(yellow) 2}{*(yellow) 2}{*(yellow) 1}{*(yellow) 1},
\none\none\none\none {*(green) 1}{*(green) 1}{*(green) 1}{*(green) 2}{*(red) 2}{*(yellow) 2}{*(yellow) 1}{*(yellow) 1}{*(yellow) 1},
\none\none\none\none\none\none{*(green) 1}{*(green) 1}{*(green) 2}{*(red) 2}{*(red) 2}{*(red) 1}{*(red) 1}{*(red) 1},
\none\none\none\none\none\none\none{*(green)1}{*(green) 1}{*(red) 2}{ 1}{ 1}{ 1},
\none\none\none\none\none\none\none\none{*(green)1}{*(red) 2}{ 1}{ 1},
\none\none\none\none\none\none\none\none{*(green)1}{*(red) 1}{ 1},
\none\none\none\none\none\none\none\none{*(green)1}{*(red) 1}
}
\end{center}
\begin{center}
=\ytableaushort{
\none\none\none{*(pink) 1},
\none\none{*(pink) 1} {*(pink) 1}, 
\none{*(pink) 1}{*(pink) 1}{*(pink) 2} 
}\quad + \ 
\ytableaushort{
{*(red) 1}{*(red) 1}{*(red) 2}{*(red) 2}{*(red) 3}{}{}{}{}{}, 
\none {*(green) 1}{*(green) 1}{*(green) 2}{*(green)  2}{*(red) 3}{}{}{}{},
\none\none{*(green) 1}{*(green) 1}{*(green) 2}{*(green) 2}{*(red) 3}{}{}{},
\none\none\none{*(green)1}{*(green) 1}{*(green) 2}{*(green) 2}{*(red) 2}{}{},
\none\none\none\none {*(green) 1}{*(green) 1}{*(green) 1}{*(green) 2}{*(red) 2}{},
\none\none\none\none\none\none{*(green) 1}{*(green) 1}{*(green) 2}{*(red) 2},
\none\none\none\none\none\none\none{*(green)1}{*(green) 1}{*(red) 2},
\none\none\none\none\none\none\none\none{*(green)1}{*(red) 2},
\none\none\none\none\none\none\none\none{*(green)1}{*(red) 1},
\none\none\none\none\none\none\none\none{*(green)1}{*(red) 1}
}\quad + \quad
\ytableaushort{
{*(red) 1},
 {*(red) 1} {*(yellow) 1},
{*(red) 2} {*(yellow) 1} {*(yellow) 1},
{*(red) 2}{*(yellow) 2}{*(yellow) 1}{*(yellow) 1}{*(yellow) 1},
{} {*(yellow) 2}{*(yellow) 2}{*(yellow) 1}{*(yellow) 1},
{}{}{*(yellow)2}{*(yellow)2}{ *(yellow)1}{*(yellow) 1},
{}{}{}{*(yellow) 2}{*(yellow) 2}{*(yellow) 1}{*(yellow) 1},
{}{}{}{}{*(yellow) 2}{*(yellow) 2}{*(yellow) 1}{*(yellow) 1},
{}{}{}{}{}{*(yellow) 2}{*(yellow) 1}{*(yellow) 1}{*(yellow) 1},
{}{}{}{}{}{}{*(red) 2}{*(red) 1}{*(red) 1}{*(red) 1}
} +\quad
\ytableaushort{
{ 1}{ 1}{ 1},
{ 1}{ 1},
{ 1}
}
\end{center}
\caption{A random configuration of reduced height function for $\ell=6$ and volume $120$. It is split into two PP's and two skew PP's.} \label{splitconfig}
\end{figure}

\subsection{Exact enumeration}

One would like to ask if there is a simple formula for $a_\ell(n)$ or for its generating function. The first few numbers for width $\ell\leq 5$ were computed by Ben Young and posted in the domino forum  \cite{Propp2014}. We adapted an algorithm due to Bratley and McKay to directly enumerate $a_\ell(n)$. Our initial numbers agree with Young's enumeration. Table \ref{exacttable} in Appendix A provides the the results of our exact enumeration of $a_\ell(n)$  for widths $\ell=1$ to $\ell=6$.

\subsubsection{The $\ell=1$ counting}

There is a natural action of the dihedral group, $D_8$, that is generated  by a rotation by $\frac\pi2$ and a reflection $(x,y)\rightarrow (-x,y)$ in the $xy$=plane. Below we indicate all possible configurations with fixed volume $n=4$ up to an overall action of $D_8$. Every point in $\mathbb{Z}^2$ is represented by a square whose entry is the reduced height at the point. The red square is the origin with the horizontal line the $x$-axis and the vertical line the $y$-axis.
\begin{center}
\begin{ytableau}
 *(red)1 \\ 1 \\  1 \\  1
\end{ytableau}
\qquad
\begin{ytableau}
1 &  *(red)1 & 1 &  1
\end{ytableau}
\qquad
\begin{ytableau}
1 \\  *(red)1 & 1 &  1
\end{ytableau}
\qquad
\begin{ytableau}
1 & 1 \\ *(red)1 & 1 
\end{ytableau}
\qquad
\begin{ytableau}
1  \\ *(red)1 & 1 \\ 1 \\
\end{ytableau}
\qquad
\end{center}
The mulitplicities of the above configurations, (from left to right), under the action of $D_8$ are $4,4,8,4,4$ respectively. Thus there are $24$ configurations with volume equal to $4$. We are interesting in counting the number of configurations with fixed volume $n$. Let $a_1(n)$ denote the number of such configurations. The first few numbers are
\begin{equation*}
1, 4, 10, 24, 51, 109, 222, 452, 890, 1732, 3298, 6204, 11470, 20970, 37842, 67572, \ldots
\end{equation*}
Let $A_\ell(q)=1+\sum_{m=1}^\infty a_\ell(n)q^n$ denote the generating function of the series $a_\ell(n)$, for fixed $\ell$. For $\ell=1$, one has
\begin{align}
A_1(q) &:= 1 + \sum_{n=1}^\infty a_1(n) \ q^n = 1+ q+4q^2+10q^3+24 q^4+\cdots\ ,\\
&=\prod_{m=1}^\infty (1-q^m)^{-b_1(m)}\ ,
\end{align}
where the second line defines $b_1(m)$ for $m=1,2,\ldots$. We have determined $b_1(m)$ for $m\leq60$. The first few numbers are:
\begin{equation}
1, 3, 6, 8, 9, 3, 2, 5, 28, 63, 86, 39, \mathbf{-112, -303, -326}, 109, 1020, 1725, 818,\ldots 
\end{equation}
If all $b_1(m)\geq0$, then one can look for a combinatorial problem that determines $b_1(m)$, thereby determining $A_1(q)$.
However, we see that $b_1(m)$ is not always positive -- the negative terms have been shown in boldface above. This behaviour is similar to what happens for solid partitions where the 
analog of $b_1(m)$ also oscillates between positive and negative values. We suspect that there might be \textbf{no} simple formula for the generating function. A similar situation holds for widths $\ell>1$. 

\subsection{Studying $a_\ell(n)$ for fixed values of $n$}

Given that there is no known analytical formula for the generating function, we next study the situation when $n$, the number of removed oranges, is kept fixed and study the properties as a function of $\ell$.  Using exact data, we find that the following formulae appear to hold for $\ell\geq \lceil{n/2}\rceil$. We set $a_\ell(0)\equiv 1=\binom{\ell}0$.
Using code which, for fixed $\ell$, generates the first few numbers in $a_\ell(n),$ enables us to conjecture the following using fits to the data:
\begin{align*}
a_\ell(2)&= \binom{\ell}2+4\ . \\
a_\ell(3)&= \binom{\ell}3+6\ell \text{ for } \ell\geq 2 \\
a_\ell(4)& = \binom{\ell}4+8 \binom{\ell}2 -\ell+23  \text{ for } \ell\geq 2\ ,\\
a_\ell(5) &=\binom{\ell}5+ 10 \binom{\ell}3 -2\binom{\ell}2+36\ell -14 \text{ for }\ell\geq 3\ , \\
a_\ell(6)&=\binom\ell{6}+ 12\ \binom\ell{4}-3 \binom\ell{3}+ 53 \binom\ell2 -25 \ell +132\text{ for }\ell\geq 3\ , \\
a_\ell(7)&=\binom\ell7+14\ \binom\ell5-4 \binom\ell4+74 \binom\ell3-40 \binom\ell2 +220 \ell -182\text{ for }\ell\geq 4\ , \\
a_\ell(8)&=\binom\ell8+16\ \binom\ell6-5 \binom\ell5+99 \binom\ell4-59 \binom\ell3 +345 \binom\ell2-308 \ell +858\text{ for }\ell\geq 4 \ ,\\
a_\ell(9)&=\binom\ell9+18\ \binom\ell7-6 \binom\ell6+ 128\binom\ell5-82 \binom\ell4 +515 \binom\ell3-488 \binom\ell2 +1463 \ell -1764\text{ for }\ell\geq 5\ .
\end{align*}
For $n=2,3,4$, the formulae have been proved  \cite{VarshaThesis}. The counting is fairly elaborate and does not reflect the simplicity of the above formulae. It hints at the existence of a statistic that refines $a_\ell(n)$ but we have been unable to find one. The na\"ive guess that it counts the number of layers affected by a given configuration does not work. For $5\leq n\leq 9$, the above formulae have been checked to be consistent with exact numbers given in Table 2 for $\ell\leq 20$. Observing their pattern, we conjecture that the following statement holds.

\begin{conj}\label{conj1}
For fixed $n$ and $\ell\geq \lceil{n/2}\rceil$, $a_\ell(n)$ is a polynomial of degree $n$ in $\ell$ such that
\begin{equation}
a_\ell(n)=\sum_{k=0}^\infty g_k(n)\  \binom{\ell}{n-k} \ , 
\end{equation}where $g_k(x)$ is a polynomial of degree $\left\lfloor \frac{k}2\right\rfloor$ in $x$.
\end{conj}
The first nine values of $n$ enables us to determine some of the polynomials to be as follows:
\begin{multline}
 a_\ell(n)=\tbinom\ell{n}+ 2n\ \tbinom\ell{n-2}-(n-3) \tbinom\ell{n-3}  +(2n^2-5n+11)\tbinom\ell{n-4}\\
 -(2n^2-11n+19)\tbinom\ell{n-5}+\frac16 ( 8 n^3 - 57 n^2+ 253 n-402 ) \tbinom\ell{n-6}+\cdots \ ,
 \end{multline}
 with $\binom{\ell}{x}=0$ for $x<0$.

\section{Asymptotics of $a_\ell(n)$}

As we have seen, it appears that we cannot come up with  a simple formula for the generating function for $a_\ell(n)$. With this in mind, we study their behaviour at large $n,$ keeping the width $\ell$ fixed. We first establish that for $\ell\ll n^{1/3}$ and $n\rightarrow \infty$ that  $n^{-2/3}\ \log a_\ell(n)$ is bounded. The proof follows a method similar to the one used to bound higher dimensional partitions  \cite{Bhatia:1997}. We then use Monte Carlo simulations to study the asymptotic behaviour more precisely.

\subsection{Bounds  on $a_\ell(n)$ for fixed $\ell$}

\begin{prop} \label{monotone} For $n\geq2$, the inequality, $a_\ell(n)> a_\ell(n-1)$, holds.
\end{prop}
\begin{proof}
Pick a configuration, $\lambda$, with volume $n$ and let $x>0$ be the largest value of $y$ such that $r(y+\ell-1,\ell-1)=1$.  If by setting $r(x+\ell-1,\ell-1)=0$, we obtain a valid configuration with volume $(n-1)$, we say that $\lambda$ has a removable $1$-part located at $(x,0)$. If $\lambda$ has a removable 1-part, then setting $r(x+\ell-1,\ell-1)=0$ corresponds to removing the $1$-part.  For example, for $\ell=1$,
\begin{ytableau}
1 &  *(red)1 & 1 &  1
\end{ytableau} has a removable $1$-part at $(2,0)$ while \ytableaushort{111{*(red) 1}}\   has no removable $1$-part. For $n>1$, adding a $1$-part to every configuration with volume $(n-1)$ generates  all configurations with volume $n$ with a removable $1$-part. 
Thus, one has
\begin{equation}
a_\ell(n) = a_\ell(n-1) + a_\ell (n| \text{no removable $1$-part})\ > a_\ell (n-1)\text{ for  } n\geq2\ .
\end{equation}
\textbf{Remarks:} Given a configuration of volume $(n-1)$, it is always possible to add a removable one-part to create a unique configuration of volume $n$ that has a removable one-part. For every $n>1$, there exists at least one configuration without a removable one-part. Consider a configuration with $r(x+\ell-1,\ell-1)=0$ for all $x>0$ and $r(\ell-1,\ell-1)=1$.
(This proof has been adapted from a proof showing that $p(n)>p(n-1),$ where $p(n)$ is the number of partitions of $n,$ given in  \cite[see chap. 3]{Andrews2004}.)
\end{proof}

\begin{figure}[htb]
\begin{center}
\includegraphics[height=2in]{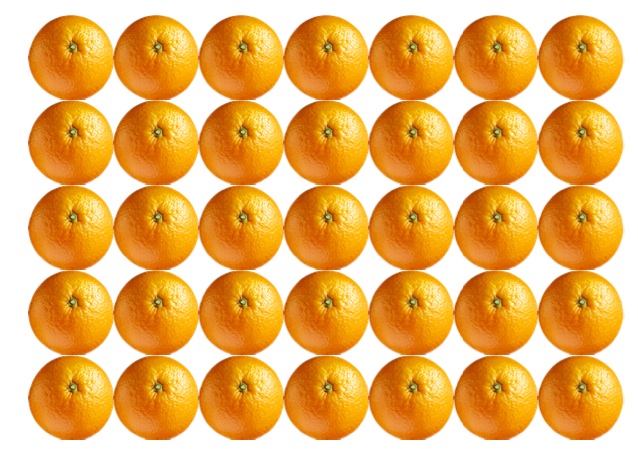}
\end{center}
\caption{The fifth layer of a stack with $\ell=3$.} \label{orangestacklayer}
\end{figure}

\begin{prop} As $n\rightarrow \infty$ and $\ell\ll n^{1/3}$, one has $\log a_\ell(n)> c_L\ n^{2/3}$ where $c_L=3^{2/3}\log2\approx 1.44$.
\end{prop}
\begin{proof}
Consider the following special configuration with $m$ layers (of oranges) completely removed. The $(m+1)$-th layer consists of $n_0=(m+1)(m+\ell)$ oranges that can all be removed independently of each other. By removing some or all of the oranges in the $(m+1)$-th layer, one creates $2^{n_0}$ configurations with volume in the range $[n-n_0, n]$ where 
$$
n=\sum_{k=1}^{m+1} k (k+\ell -1) = \tfrac16 (1 + m) (2 + m) (2 m+3\ell)\ .
$$ 
We express $m$ in terms of $n$ by inverting the above expression to obtain \[m= (3n)^{1/3} -\frac{(\ell+2)}2 +O(n^{-1/3})\ .\] Similarly, we can see that $n_0 = (3n)^{2/3} + (3n)^{1/3} + O(1)$. 
 Since these $2^{n_0}$ configurations  do not exhaust all possible configurations, one has
 \begin{equation*}
 \sum_{n'=n-n_0}^n a_\ell(n') > 2^{n_0}\ .
 \end{equation*}
 Since $a_\ell(n'+1)> a_\ell(n')$ for $n'>1$ from Proposition \ref{monotone}, we obtain
 \begin{equation*}
n_0\ a_\ell(n)>  \sum_{n'=n-n_0}^n a_\ell(n') > 2^{n_0}\ .
 \end{equation*}
We thus get the following lower bound
 \begin{align}
 \log a_\ell(n) & > (\log 2)\ n_0 -\log n_0  \nonumber \\ 
 &> (\log 2)\ n_0 = (\log 2) (3n)^{2/3} + O(n^{1/3}) =: c_L \ n^{2/3}+ O(n^{1/3})  \ ,
\end{align}
with $c_L=3^{2/3}\log2\approx 1.4418$.
\end{proof}

\begin{prop}
As $n\rightarrow \infty$ and $\ell\ll n^{1/3}$, one has $\log a_\ell(n) < c_U\ n^{2/3}$ where $c_U = 3\zeta(3)^{1/3}\approx 3.1898$.
\end{prop}
\begin{proof}
Let $p_2(n)$ denote the number of plane partitions of $n$ and $\hat{p}_2^{(\ell)}(n)$ denote the number of skew plane partitions of shape $\lambda /\mu_\ell$,  where $\mu_\ell$ is the Ferrers diagram for  partition $(\ell,\ell-1,...,1)$ and $\lambda$ the Ferrers diagram of a partition containing $\mu_\ell.$
We obtain the following upper-bound for $n\gg 1$.
\begin{align*}
a_\ell(n) & <  \sum_{\substack{n_i\in \mathbb{Z}_+\\\sum_i n_i = n}}\  \ \prod_{j=1}^2 p_2(n_j) \prod_{j=3}^4 \hat{p}^{(\ell)}_2(n_j)\ , \\
&<  \sum_{\substack{n_i\in \mathbb{Z}_+\\\sum_i n_i = n}}\   \prod_{j=1}^2 p_2(n_j) \ \prod_{j=3}^4 p_2(n_j+\tfrac{\ell^2}{2} (3n)^{1/3})\ , 
\end{align*}
where in the second line, we have replaced the counting of skew plane partitions to plane partitions by filling in $\mu_\ell$ with the largest possible value which can be estimated to be $(3n)^{1/3}$.
Since $\ell \ll n^{1/3}$, we assume that it is $O(1)$. Since $p_2(n)$ is a monotonically increasing function of $n$, it follows  that among all partitions of $n$ into four parts, the largest term in the  above product occurs when all $n_i$ are equal. Thus one has $\ell^2 n^{1/3} \ll n_j$ for $j=3,4$. Using this, we obtain
\begin{align}
a_\ell(n) < \  p(n|4\text{ parts})\  p_2\!\left(\tfrac{n}4\right)^4\ .
\end{align}
where $p(x|4\text{ parts})=O(x^3)$ is the number of partitions of $x$ into four parts. 
Taking logarithms and discarding terms that grow as $\log n$ that arise from $p(n|4\text{ parts})$, we obtain
\begin{equation}
\log a_\ell(n) < 4 \log p_2\left(\tfrac{n}4\right) \sim 3 \zeta(3)^{1/3}  \ n^{2/3} = 3.1898\ n^{2/3} \ ,
\end{equation}
on using $\log p_2(n) \sim \tfrac32 (2\zeta(3))^{1/3}\ n^{2/3},$ see  \cite{Wright1931}.
\end{proof}
Combining our lower and upper bounds, we obtain  the following bounds:
\begin{equation}
\boxed{
3^{2/3}\log2<  n^{-2/3}\ \log a_\ell(n) < 3 \zeta(3)^{1/3} \ .
}
\end{equation}
This suggests that $n^{-2/3}\log a_\ell(n)\rightarrow \text{constant}$ as $n\rightarrow \infty$.
\begin{conj}
For $\ell\ll n^{1/3}$, $n^{-2/3} \log a_\ell(n)\sim $ an $\ell$-independent constant as $n\rightarrow\infty$.
\end{conj}
A heuristic proof of $\ell$-independence is as follows. Since $\ell\ll n^{1/3}$, arguments similar to those that lead to the lower bound show that a generic random configuration will be a rectangle of side $(3n)^{1/3} \big[(3n)^{1/3}+\ell]\sim (3n)^{2/3} +\ell\ O(n^{1/3})$. This suggests that the $\ell$-dependence is suppressed by at least a power of $n^{1/3}$. We shall provide evidence for this using Monte Carlo simulations to estimate the constant for $\ell=1,\ldots,6$.

\subsection{Studying asymptotics using Monte Carlo simulations}

Let $\lambda$ denote a particular height function (or equivalently a stack of oranges) with volume $n$. We indicate this by $\lambda \vdash n$.  Let $n_+(\lambda)$ ($n_-(\lambda)$) denote the number of oranges that can be removed (resp. added) to obtain a valid height function with volume $(n+1)$ (resp. $(n-1)$). 
 Define $N_\pm(n)$ as follows:
\begin{align}
N_+(n):=&\frac{\sum_{\lambda\vdash n} n_+(\lambda)}{\sum_{\lambda\vdash n} 1 } =
\frac{\sum_{\lambda\vdash n} n_+(\lambda)}{a_\ell(n) } \quad \textrm{and}\quad \nonumber \\
N_-(n):=&\frac{\sum_{\lambda\vdash n} n_-(\lambda)}{\sum_{\lambda\vdash n} 1 } =
\frac{\sum_{\lambda\vdash n} n_-(\lambda)}{a_\ell(n) } \ ,
\end{align}
where the sums run over all height functions with volume $n$. For $n>1$, one has the identity
\begin{equation}
N_+(n-1)\ a_\ell(n-1) = N_-(n)\ a_\ell(n)\ .
\end{equation}
Given $N_+(n)$ and $N_-(n)$, one can determine $a_\ell(n)$ by recursively using the formula and using $a_\ell(0)=1$. That is,
\begin{equation}
a_\ell(n) = \prod_{m=0}^{n-1} \frac{N_+(m)}{N_-(m+1)}\ , \label{aldeterminea}
\end{equation}
or for $n>n_0$ (where $a_\ell(n_0)$ has been exactly enumerated)
\begin{equation}
a_\ell(n) = \prod_{m=n_0}^{n-1} \frac{N_+(m)}{N_-(m+1)}\ a_\ell(n_0)\ , \label{aldetermineb}
\end{equation}

The transition matrix Monte Carlo simulation we use estimates averages for $N_\pm(n)$ for $n\in [1,4100]$ for $\ell=1,\ldots,6$. We assume that $\log a_\ell(n)$ takes the following asymptotic form:
\begin{equation}
\log a_\ell(n) \sim c_0\  n^{2/3} + c_1\ \log n  + c_2  + c_3\ n^{1/3} \ .
\end{equation}
Using this form, one can show that
\begin{equation}
\log \frac{a_\ell(n)}{a_\ell(n-1)} =  \log \frac{N_+(n-1)}{N_-(n)} \sim \frac23\ c_0\ n^{-1/3}  + c_1 \ n^{-1} + \frac13\ c_3\ n^{-2/3} \ . 
\end{equation}
For our Monte Carlo fits, we use a variant of the above formula
\begin{equation}
\log \frac{a_\ell(n)}{a_\ell(n-1)}  \sim \left(\tfrac23+\tfrac{1}{9n}\right)\ c_0\ n^{-1/3}  + c_1 \ n^{-1} + \left(\tfrac13+\tfrac{1}{9n}\right)\ c_3\ n^{-2/3} \ ,
\label{npmpars}
\end{equation}
where we have added some sub-leading terms (suppressed by $1/n$) without changing the number of parameters.
This formula is suited to our Monte Carlo simulation as it relates the quantities computed in the simulation to the parameters that appear in
the asymptotic form for $a_\ell(n)$. The parameter $c_2$ has to be determined separately as it drops out of the above formula.


The Monte Carlo simulation is a randomisation of the Bratley-McKay algortihm.  We adapted the Transition Matrix Monte Carlo method described in  \cite{Widom:2002} to study solid partitions restricted to be in a box and to estimate the asymptotics of solid partitions in  \cite{Destainville2014}. As in those papers, we use a fictitious temperature to get a wider coverage for values of $n\in[1,N_\text{max}]$. The averages for estimating $N_\pm(n)$ are carried out at infinite temperature. We carried out several runs with different values of 
$N_\text{max}=  1200, 2200, 4200, 10200$. For each value of $N_\text{max}$, we carried out runs with distinct seeds for the random number generator in order to get an estimate of the statistical error in $N_\pm(n)$. The numbers from all runs were then combined into a single data set with statistical errors. For $n\in[1,30]$, the values of $N_\pm(n)$ were compared with exact values (again computed using the Bratley-McKay algorithm  \cite{VarshaThesis}) to see if the statistical errors that we obtained were consistent with actual ones. The exact numbers also enabled us to establish that longer runs lead to lower statistical errors. As a proof of concept, we also verified that a similar randomisation of the Bratley-McKay code for ordinary partitions worked. The runs with $N_\text{max}=10200$ were not used in any of our fits as their errors were too large and were only used to verify that our fits do reproduce the asymptotic behaviour correctly.

\begin{figure}
\centering
\includegraphics[width=3in]{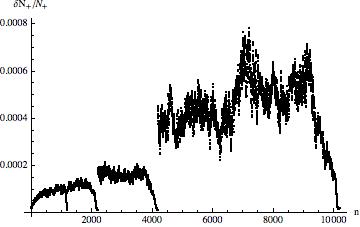}
\caption{Plot of statistical error, $\frac{\delta N_+}{N_+}$, against $n$. The merger of four data sets is also clearly visible. All statistical errors show similar behaviour.}
\end{figure}

\subsubsection{Summary of Monte Carlo results}
We carried out three sets of fits using estimates for $N_\pm(n)$  for values of $n$ in the range $[\ell^3+10, 4100]$.  The data for $n>4100$ has larger errors and hence is only used to see if the trends of the fits are consistent.
\begin{enumerate}
\item[\bf Fit 1:]
This is the formula given in Eq. \eqref{npmpars} which involves three parameters. The first fit gives
\begin{center}
\begin{tabular}{c|c|c|c}
$\ell$ & $c_0$ & $c_3$ & $c_1$  \\ \hline
1 &$2.34426$ &  $-0.0110902$ &  $-0.746477$ \\
2& $2.34437$ & $-0.0156179$ &  $-0.740064$ \\
3 & $2.34441$ & $-0.0281878$ & $-0.670066$ \\
4 &$2.34492$ & $-0.0669717$ & $-0.5053$ \\
5 & $2.34558$ & $-0.119494$ & $-0.248828$ \\
6 &$2.34538$ &  $-0.144212$ &  $0.0401944$ \\ \hline
\end{tabular}
\end{center}
\item[\bf Fit 2:] The second fit is one where a
fourth parameter is introduced by adding a term $\varepsilon\ n^{-4/3}$ to the right hand side of Eq. \eqref{npmpars}. 
The second fit gives
\begin{center}
\begin{tabular}{c|c|c|c|c}
$\ell$ & $c_0$ & $c_3$ & $c_1$  & $\varepsilon$ \\ \hline
1 &$2.34401$ &  $0.0028884$ & $-0.78056$ & $0.0630788$ \\
2& $2.34417$ &$0.00278436$ & $-0.770915$ & $0.064166$ \\
3 &$2.3379$ & $0.0104204$ & $-0.783435$ & $0.277578$ \\
4 &$2.34397$ & $-0.00361808$ &$-0.712328$& $0.589102$ \\
5 & $2.34444$ &  $-0.0387065$ & $-0.538716$ & $0.935576$ \\
6 &$2.34329$ &  $-0.0140978$ & $-0.575709$  & $2.2077$ \\ \hline
\end{tabular}
\end{center}
\item[\bf Fit 3:] 
A third form for the asymptotic behaviour, based on the (leading) singularity of the generating function, is \[a_\ell(n) \sim A\ \mu^{n^{2/3}}\ n^g\ .\] 
Comparing with the first asymptotic formula, we see that $A=e^{c_2}$, $\mu = e^{c_0}$, $g=c_1$  and  $c_3=0$.
For the third fit we also added the term $\varepsilon\ n^{-4/3}$ term, giving 
\begin{center}
\begin{tabular}{c|c|c|c|c}
$\ell$ & $c_0$ & $\mu=e^{c_0}$&  $g=c_1$  & $\varepsilon$ \\ \hline
1 & $ 2.34407$&$10.4236$ & $-0.7741296$&$ 0.0520715$ \\
2&$2.34412$&$10.4241$ & $ -0.777704$&$ 0.0773793$ \\
3&$2.34397$&$10.4225$&$-0.754525$&$ 0.210345$ \\
4&$2.34391$&$10.4219$&$ -0.723663$&$ 0.620114$\\
5& $2.34389$&$10.4217$&$ -0.673265$&$ 1.35694$ \\
6&$ 2.34348$&$10.4174$&$-0.522174$&$ 2.02021 $ \\ \hline
\end{tabular}
\end{center}
We see that forcing $c_3=0$ makes the value of $c_0$ almost independent of $\ell$ providing evidence to our conjecture that $c_0$ is $\ell$-independent. We assign it the $\ell$-independent  value 
\begin{equation}
\boxed{
c_0=2.344\pm 0.001 \text{ or } \mu=10.42\pm0.01\ .
}
\end{equation} 
The errors here are crude estimates based on comparing how the numbers change when compared to the second fit. Further the parameter $g=c_1$ is clearly $\ell$-dependent. 
\end{enumerate}
\begin{figure}
\centering
\includegraphics[height=2in]{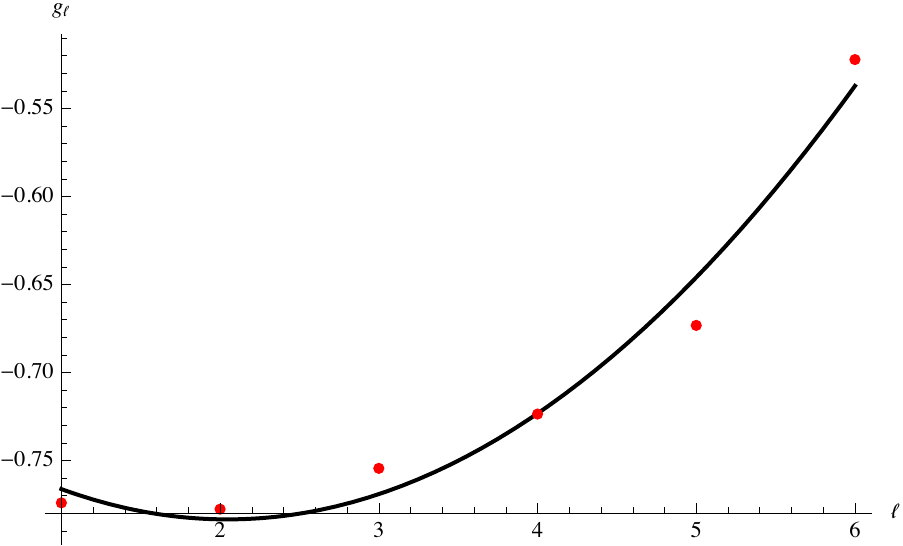}
\caption{$g_\ell$ vs $\ell$ along with a quadratic fit which gives $g_\ell=-0.717604 - 0.0644484 \ell + 0.0157537 \ell^2$.}
\end{figure}

The main conclusion that we can draw from the Monte Carlo simulations is that the asymptotic behaviour of $a_\ell(n)$ is consistent
with the following form:
\begin{equation}
a_\ell(n) \sim A_\ell \ \mu^{n^{2/3}} \ n^{g_\ell} = e^{c^{\ell}_2}  \ \mu^{n^{2/3}} \ n^{g_\ell}\ ,
\label{formulafinal}
\end{equation}
where $\mu=10.42\pm0.01$ is an $\ell$-independent constant and $A_\ell=e^{c^{\ell}_2} $ and $g_\ell$ are $\ell$-dependent constants. 

We still need to estimate $A_\ell$ or equivalently the constant $\alpha_3^\ell$ as it does not appear in the fits based on Eq. \eqref{npmpars}.  We need explicit values for $a_{\ell}(n)$ -- this is something  we  indirectly determine using our estimates for $N_\pm(n)$  combined with Eq. \eqref{aldetermineb} with $n_0$ chosen to be the largest  possible value appearing
in our explicit enumeration given in Table \ref{exacttable}. We fit to the formula
\begin{equation}
n^{-2/3}\log a_\ell(n) \sim c_0  +c_1\ n^{-2/3}\log n  + c_2\ n^{-2/3} -3\varepsilon\ n^{-1}  \ ,
\end{equation}
with the values of $c_0$, $g_\ell$ and $\varepsilon$  determined by Fit 3.  We use small values of $n\in [\text{max}(10,\ell^3), \ell^3+100]$ as it is here that this term contributes significantly and statistical errors are small.
\begin{center}
\begin{tabular}{c|c|c}
$\ell $ & $c_2$ & $A=e^{c_2}$\\[3pt] \hline
 1 & $-1.55101$ & $0.212034$ \\
 2 & $-1.2617$ & $0.283173$\\
 3 & $-0.64815$ & $0.523012$ \\
 4 & $0.356079$ & $1.42772$\\
 5 & $1.64144$& $5.16257$ \\
 6 & $2.52126$ & $12.4442$ \\ \hline
\end{tabular} 
\end{center}

\begin{figure}[thb]
\begin{center}
\includegraphics[width=2.4in]{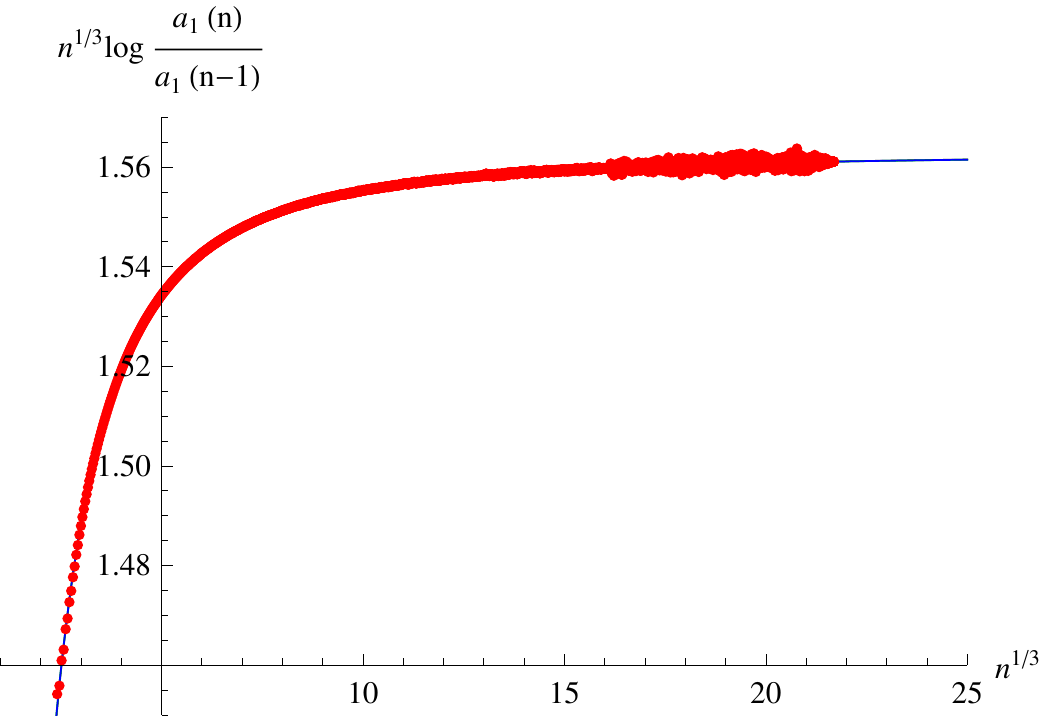} \qquad \includegraphics[width=2.4in]{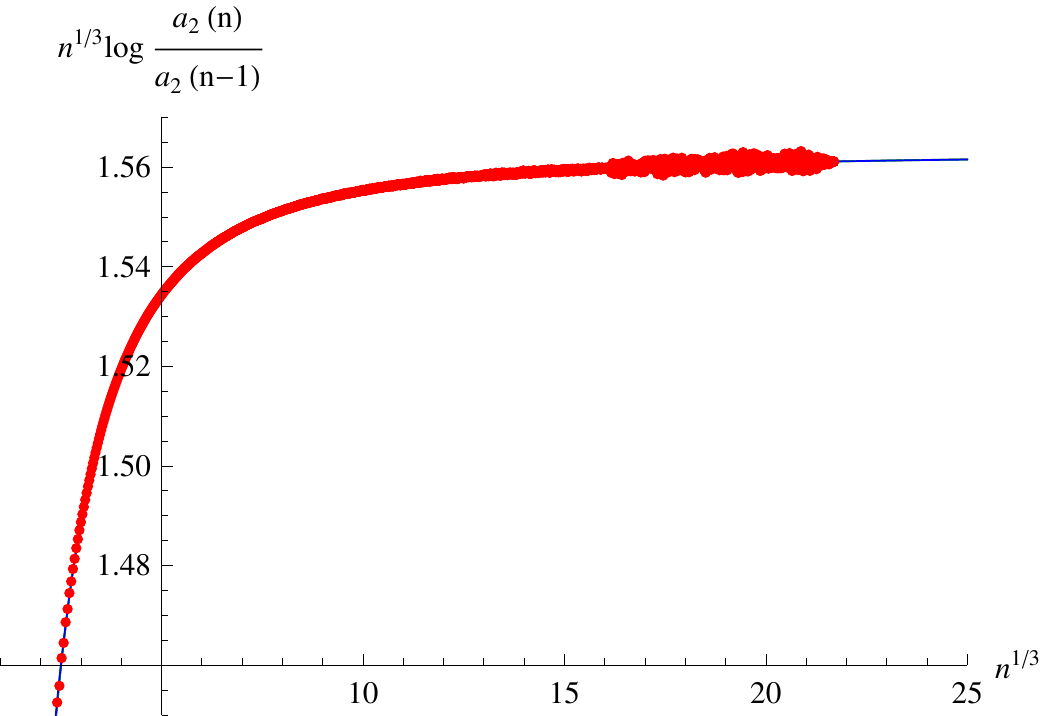} \\[5pt]
\includegraphics[width=2.4in]{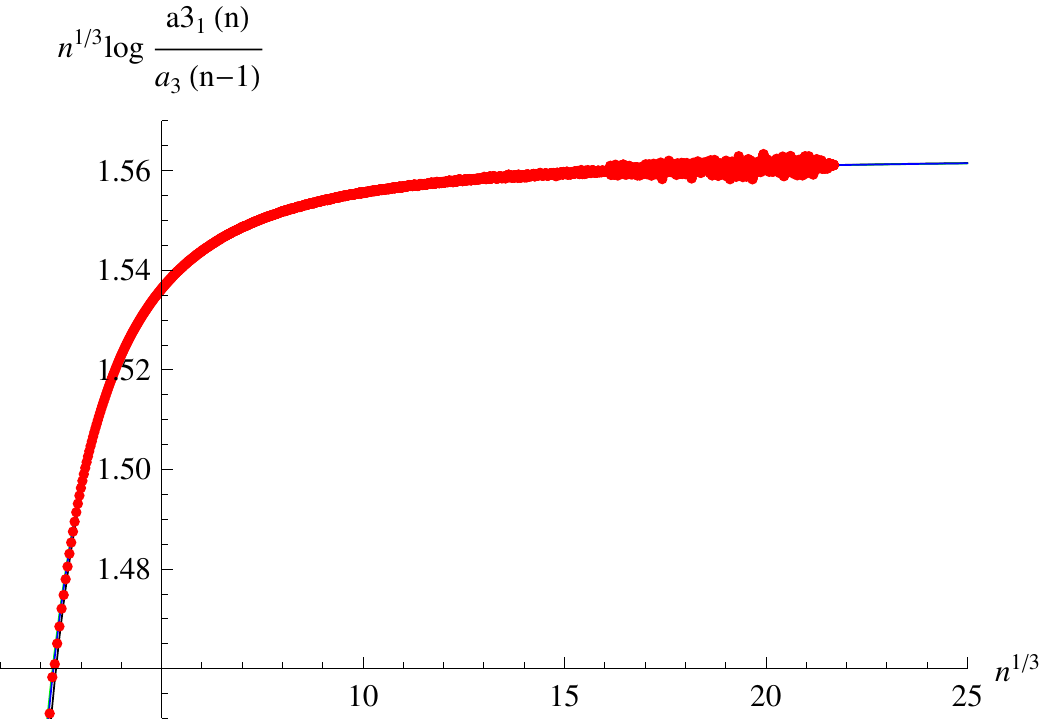} \qquad \includegraphics[width=2.4in]{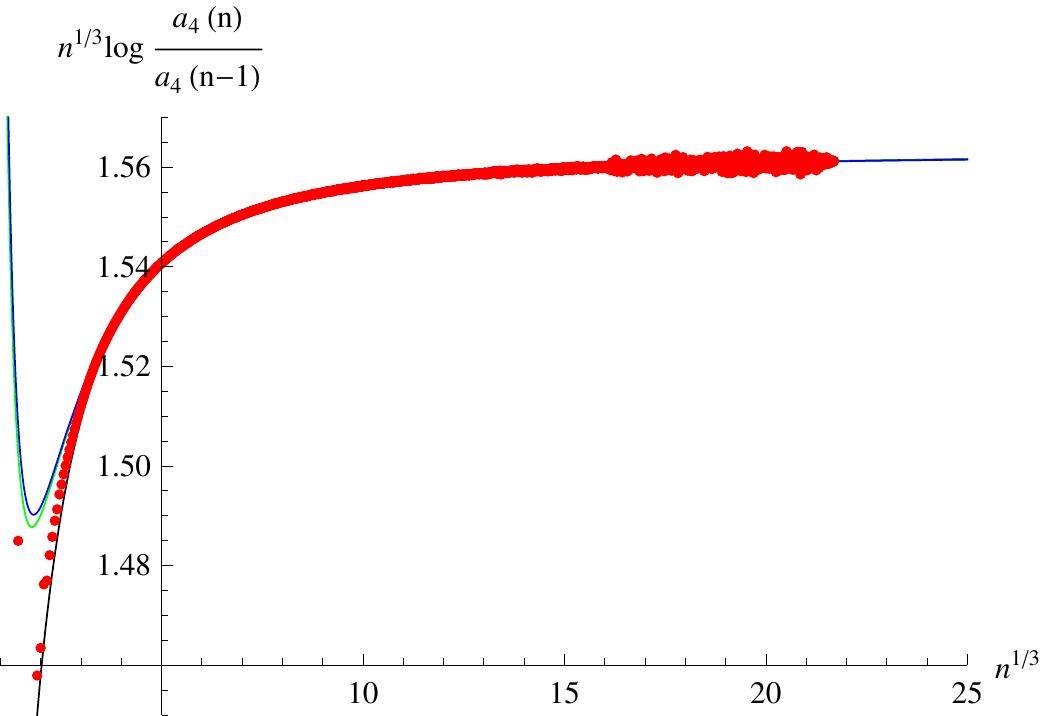} \\[5pt]
\includegraphics[width=2.4in]{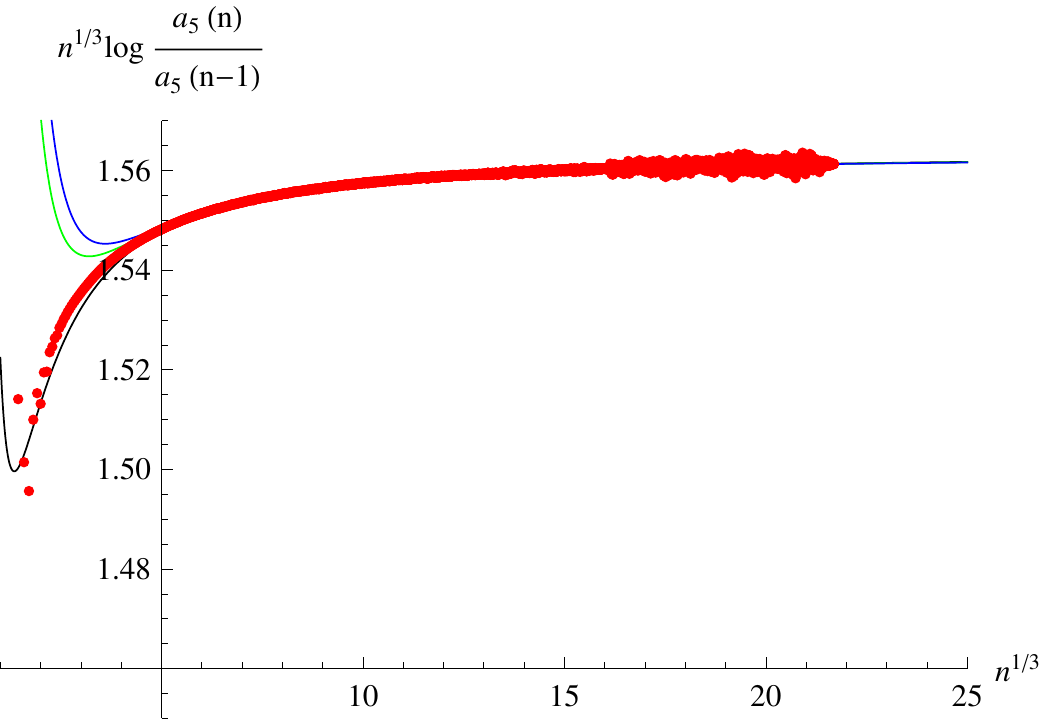} \qquad \includegraphics[width=2.4in]{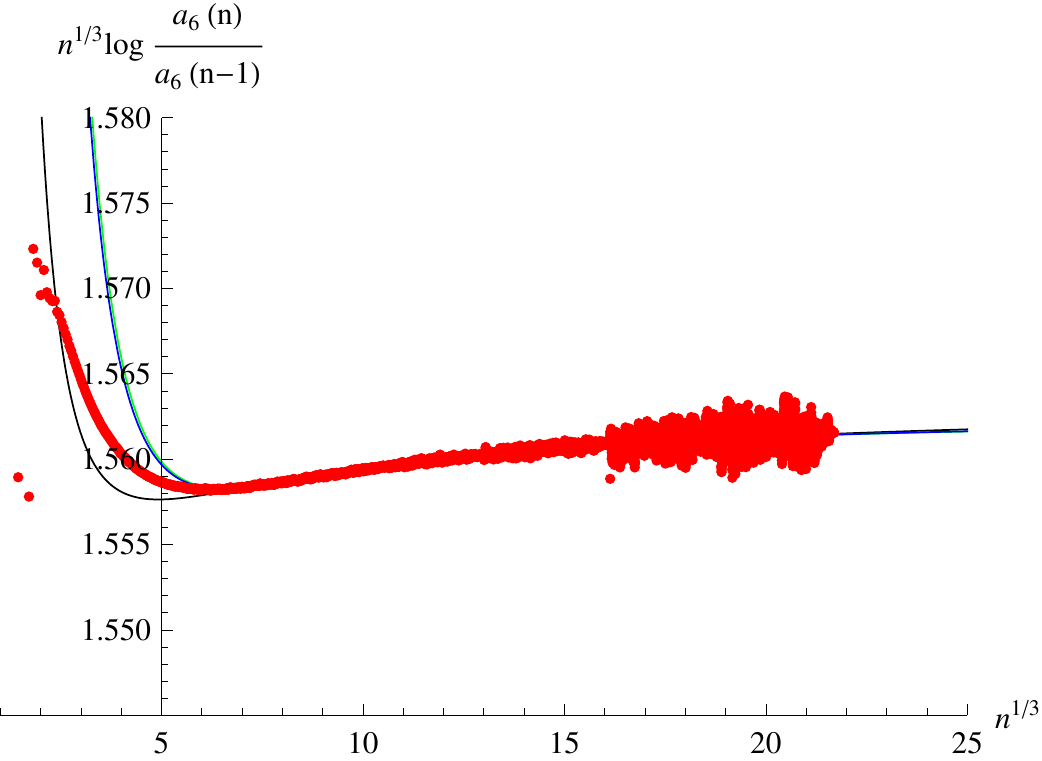}
\end{center}
\caption{Plots of $\log \frac{a_\ell(n)}{a_\ell(n-1)}$ vs $n^{1/3}$ for $\ell=1,\ldots,6$. The three fits are plotted  with fit1 in black, fit2 in green, fit3 in blue and red for the Monte Carlo data. The fits should work only for $n^{1/3}>\ell$ and the fits differ when $n^{1/3}<\ell$.} 
\end{figure}

\section{Series analysis of partition-type series}
Much of the pre-existing work on methods to extract the asymptotic form of coefficients numerically from a finite number of coefficients assumes the form
$$a_n \sim \text{const.} \mu^n n^g,$$ with corresponding generating function $$\sum a_n x^n \sim \text{const.}\ \left( 1 - \mu x \right )^{-1+g}.$$ Many problems in enumerative combinatorics and statistical mechanics have such singularities. Methods for the analysis of coefficients in order to estimate the growth constant $\mu, $ the exponent $g$ and the amplitude, given by the constant pre-multiplier have been well-developed over the past few decades, and are discussed in  \cite{G89}.

In contrast, for the type of asymptotics associated with plane partitions and related series, the literature is very scant indeed. Accordingly, we first take a known problem, the asymptotics of plane partitions, and develop appropriate methods of series analysis. We then apply these methods to the problem at hand, the square-ice analogue of plane partitions. 

\subsection{Analysis of plane-partition series}
The generating function of plane partitions, due to MacMahon \cite{Macmahon1896memoir}, is well-known and is given by $$P(x)=\sum p_n x^n =\prod_{k \ge 0} \frac{1}{(1-x^k)^k}=1+x+3x^2+6x^3+13x^4 + \cdots.$$ The asymptotics are also well-studied, and are given  by \cite{Wright1931,Mutafchiev2006}
\begin{equation}
\label{asy1}
n^{-2/3}\log{p_n} \sim c_0+c_1 \frac{\log{n}}{n^{2/3}}+\frac{{c_2}}{n^{2/3}} +O(n^{-4/3}),
\end{equation}
 where $c_0 = 2.00945\cdots,$ $c_1=-\frac{25}{36} = -0.694444\cdots,$ and $c_2 = -1.4631 \cdots.$ 

It is straightforward to generate as many terms as required from the generating function. We have chosen to generate 200 terms, and investigate the assumed form 
$${\tilde p_n} =n^{-2/3}\log{p_n} \sim c_0+c_1 \frac{\log{n}}{n^{\alpha}}+\frac{{c_2}}{n^{\alpha}},$$ with higher order terms neglected. That is to say, we assume ignorance of the exponent $\alpha,$ and set out to estimate its value. 

Forming first-differences, so that
$$s_n={\tilde p_n}-{\tilde p_{n-1}} \sim -c_1\alpha \frac{\log{n}}{n^{1+\alpha}} + O\left (\frac{{1}}{n^{1+\alpha}} \right ),$$ then a plot of $s_n$ against $\frac{\log{n}}{ n^{(1+\alpha)}}$ should be linear for the ``correct" choice of $\alpha$ and $n$ sufficiently large. This is not a particularly sensitive test, but one might expect to establish if $\alpha$ is closer to 1 or to zero. In Figure \ref{fig:1} we show such a plot for three values of $\alpha.$ For $\alpha=1,$ shown at left, the plot is slightly convex, while the right-most plot, corresponding to $\alpha=0.5$ is significantly concave, while the central plot, corresponding to $\alpha=0.75$ is essentially linear. The correct value of $\alpha$ is of course $2/3$ in this case.

\begin{figure}[htbp] 
 \centering
    \includegraphics[width=2.5in]{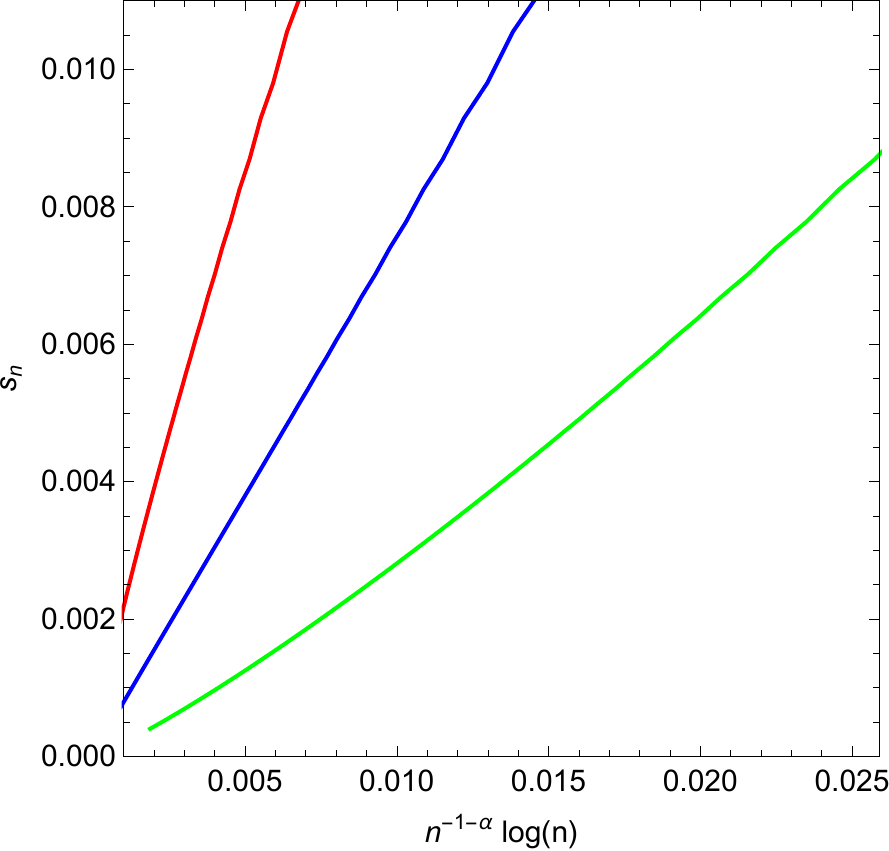} 
   \caption{Left-most plot, $\alpha=1,$ central plot, $\alpha=0.75$ and right-most plot, $\alpha=0.5.$ for plane partitions.}
   \label{fig:1}
\end{figure}

An alternative way to estimate $\alpha$ is to plot $\log \left ( \frac{s_n}{\log {n}} \right )$ against $\log{n}.$ This should have gradient $-(1+\alpha).$ This plot (not shown) is indeed visually linear. If one calculates the local gradient, defined as the gradient of successive pairs of points, one sees a steady variation with $n.$ This local gradient is plotted against $n^{-2/3}$ in the left-most plot in figure \ref{fig:g1}. It is clear that this is extrapolating to a value around $-1.68$ as $n \to \infty,$ which is quite close to the known exact value $-5/3.$ 

Assuming we have found the value of $\alpha$ correctly to be $2/3,$ we are now in a position to estimate the constants appearing in the asymptotic expression (\ref{asy1}). We fit successive triples of terms $\{\tilde{p}_{n-2},\,\,\tilde{ p}_n,\,\,\tilde{ p}_{n+2}\}$ in order to estimate the constants $\{c_0(n),\,\, c_1(n),\,\, c_2(n) \}.$ (Alternate terms are used to reduce an odd-even effect that would otherwise cause oscillatory estimates). We show the estimates of these constants, plotted against $n^{-4/3},$  $n^{-2/3}$ and $n^{-1/3}$ respectively in figures  \ref{fig:g1} and \ref{fig:c1-pp} below. The estimates of $c_0$ are clearly going to a value around $2.0095,$ which is very close to the exact value. The estimates of $c_1$ appear to be going to a limit around $-0.695$, in good agreement with the known exact value, $-0.69444\cdots$. The estimate for $c_2 \approx -1.436$  which is comparable to the known value of $c_2=-1.4631\cdots,$. 

\begin{figure}[h] 
 \centering
    \includegraphics[width=2.5in]{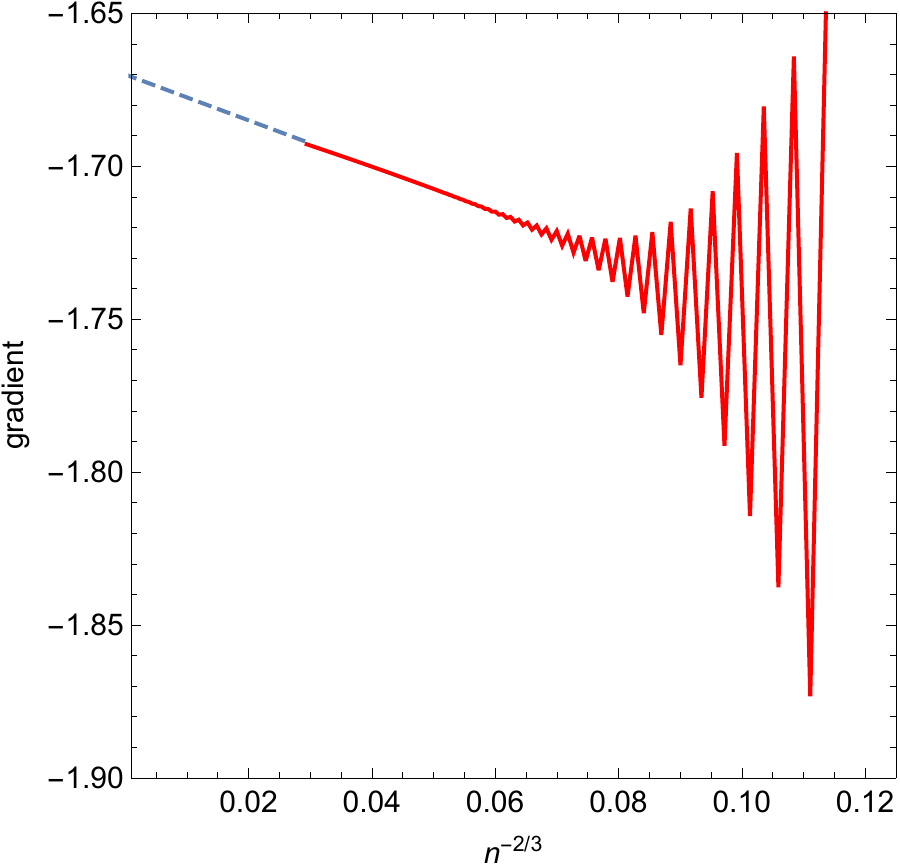}  \qquad \includegraphics[width=2.5in]{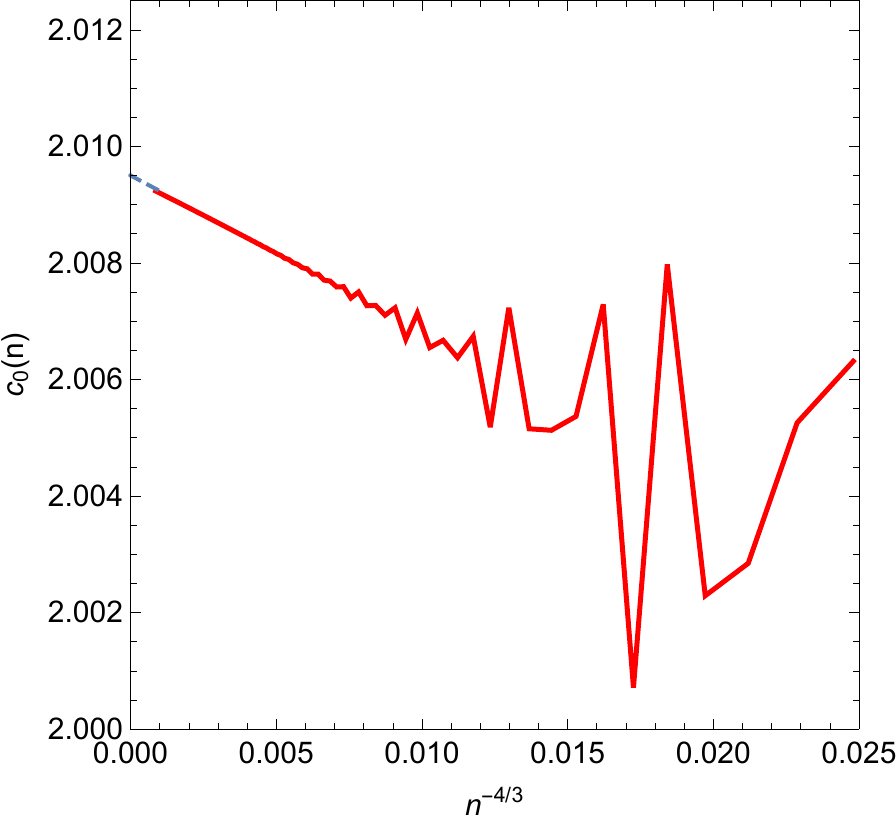} 
   \caption{(Left) Plot of local gradient against $n^{-2/3}$ and (right) Plot of $c_0(n)$ against $n^{-4/3}$ for plane partitions.}
   \label{fig:g1}
\end{figure}


\begin{figure}[h] 
    \includegraphics[width=2.5in]{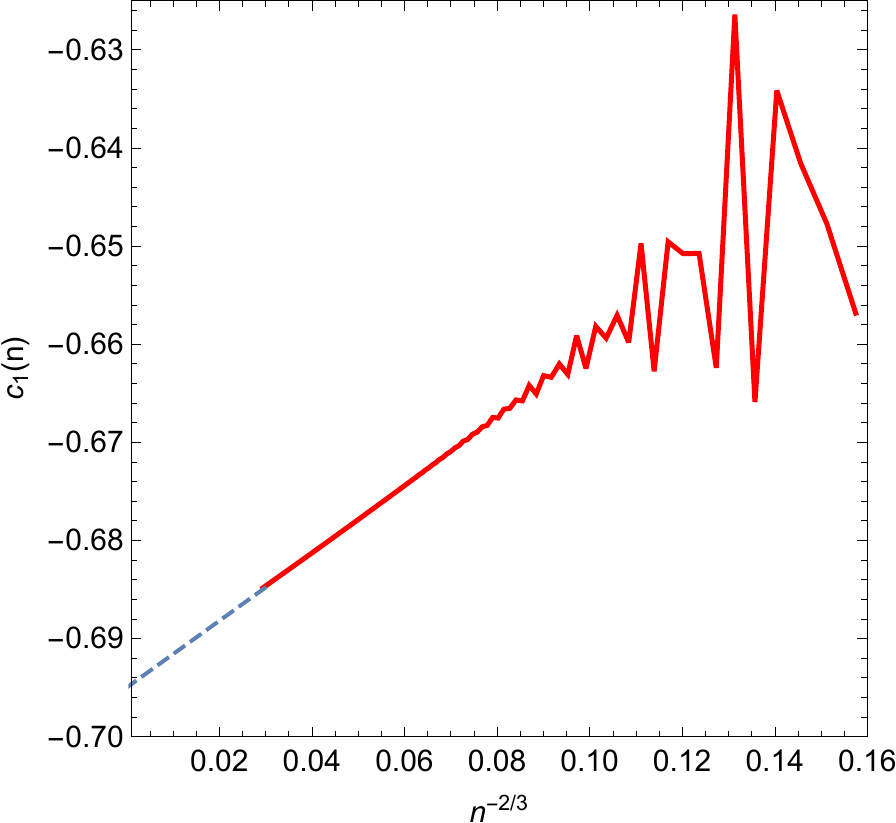} \qquad
%
   \includegraphics[width=2.5in]{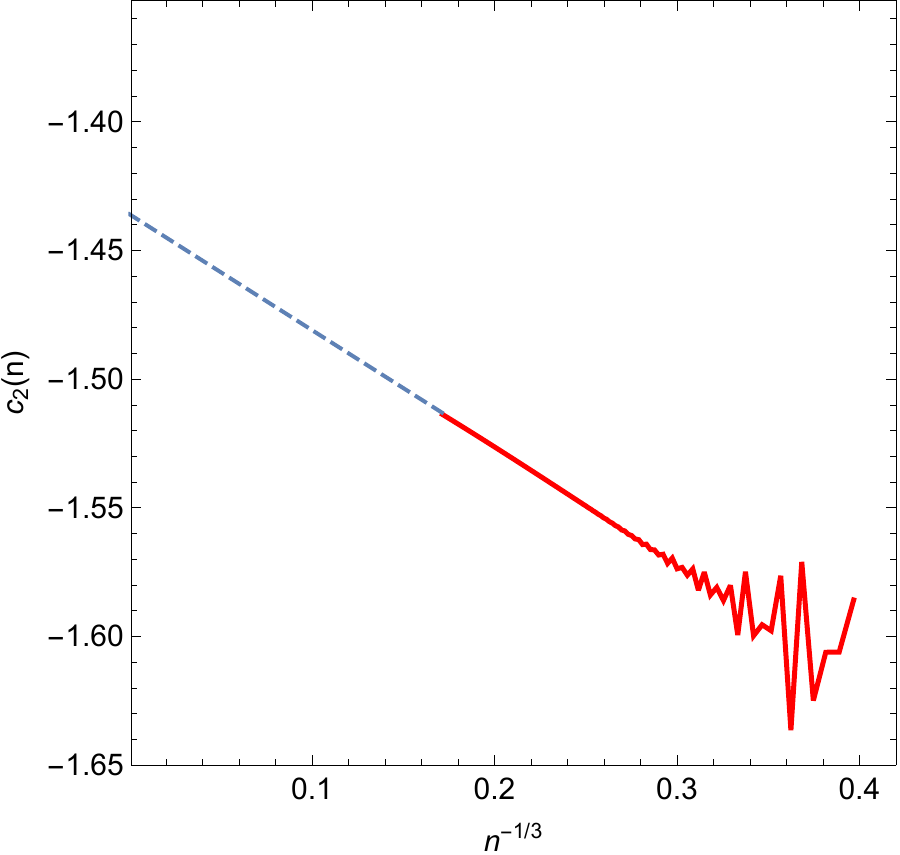} 
   \caption{Plots of $c_1(n)$ and $c_2(n)$ against $n^{-2/3}$ (resp. $n^{-1/3}$) for plane partitions.}
  \label{fig:c1-pp}
\end{figure}

We have repeated the above analysis with an additional term $c_3/n^{1/3}$ in (\ref{asy1}), and the estimators of $c_3$ are clearly going to a value close to 0, consistent with the absence of such a term.
\subsection{Analysis of square-ice series}
We now repeat the above analysis for the sequence $a_1(n)$ which is known exactly for $n \le 60.$ We have recently developed a numerical technique that allows one to approximately extend a given series by several coefficients, with a level of precision that is good enough for this type of graphical analysis, see  \cite{G16}. In this way we have extended the series by 10 further terms, and these are quoted in Table \ref{seriescompare} alongside the estimates from the Monte Carlo simulations.

As in the preceding case, we first form the sequence $${\tilde p_n} =n^{-2/3}\log{a_1(n)} \sim c_0+c_1 \frac{\log{n}}{n^{\alpha}}+\frac{{c_2}}{n^{\alpha}},$$ with higher order terms neglected, and we calculate the first-differences, $s_n={\tilde p}_n-{\tilde p}_{n-1}$ and plot $s_n$ against $\frac{\log{n}}{ n^{(1+\alpha)}}.$  We show the results in Figure \ref{fig:a1}, again for three values of $\alpha.$ The situation is exactly the same as for plane partitions. For $\alpha=1,$ shown at left, the plot is slightly convex, while the right-most plot, corresponding to $\alpha=0.5$ is significantly concave, while the central plot, corresponding to $\alpha=0.75$ is essentially linear. This suggests that the correct value of $\alpha$ is also $2/3$ in this case.

\begin{figure}[htbp] 
 \centering
    \includegraphics[width=2.5in]{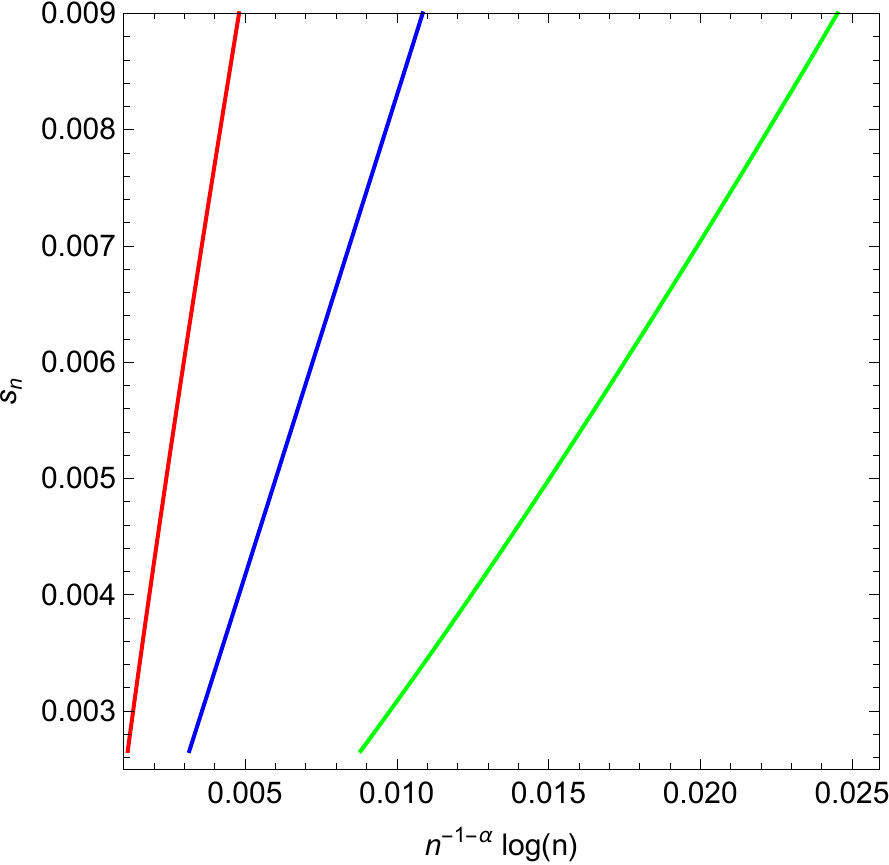} 
   \caption{Left-most plot, $\alpha=1,$ central plot, $\alpha=0.75$ and right-most plot, $\alpha=0.5.$ for the square-ice series.}
   \label{fig:a1}
\end{figure}

Estimating $\alpha$ by plotting $\log \left ( \frac{s_n}{\log {n}} \right )$ against $\log{n}$ again gives a visually linear plot. More interesting is the plot of the local gradient, and this is shown plotted against $1/n^{2/3}$ in Figure \ref{fig:a1}. This appears to extrapolate to a value around $-1.68$ as $n \to \infty,$ just as for plane partitions, which again suggests that the correct exact value should be $-5/3.$ 

Assuming we have found the value of $\alpha$ correctly to be $2/3,$ we are now in a position to attempt to estimate the constants appearing in the asymptotic expression (\ref{asy1}). As for the case with plane partitions, we fit successive triples of terms $\{\tilde{p}_{n+2},\,\,\tilde{ p}_n,\,\,\tilde{ p}_{n+2}\}$ in order to estimate the constants $\{c_0(n),\,\, c_1(n),\,\, c_2(n) \}.$  We show the estimates of these constants, plotted against suitable powers of $n$,  $\{n^{-4/3},n^{-2/3},n^{-1/3}\}$,  in Figures \ref{fig:c0-si} and \ref{fig:c1-si} below. All display oscillatory behaviour which makes extrapolation difficult, if not impossible. If we assume -- and this is indeed a leap of faith, justifiable only because the results are consistent with the Monte Carlo analysis -- that this oscillatory trend persists with decreasing amplitude, then we can estimate $c_0 \approx 2.345,$  $c_1 \approx -0.75$ and $c_2 \approx -1.7.$ 

These results are entirely consistent with, though less accurate than, the Monte Carlo estimates obtained from the third fit, which assumes $c_3$ is zero (that is, there is no term $O(n^{-1/3})$ in Eq. (\ref{asy1})).
\begin{figure}[h] 
   \includegraphics[width=2.5in]{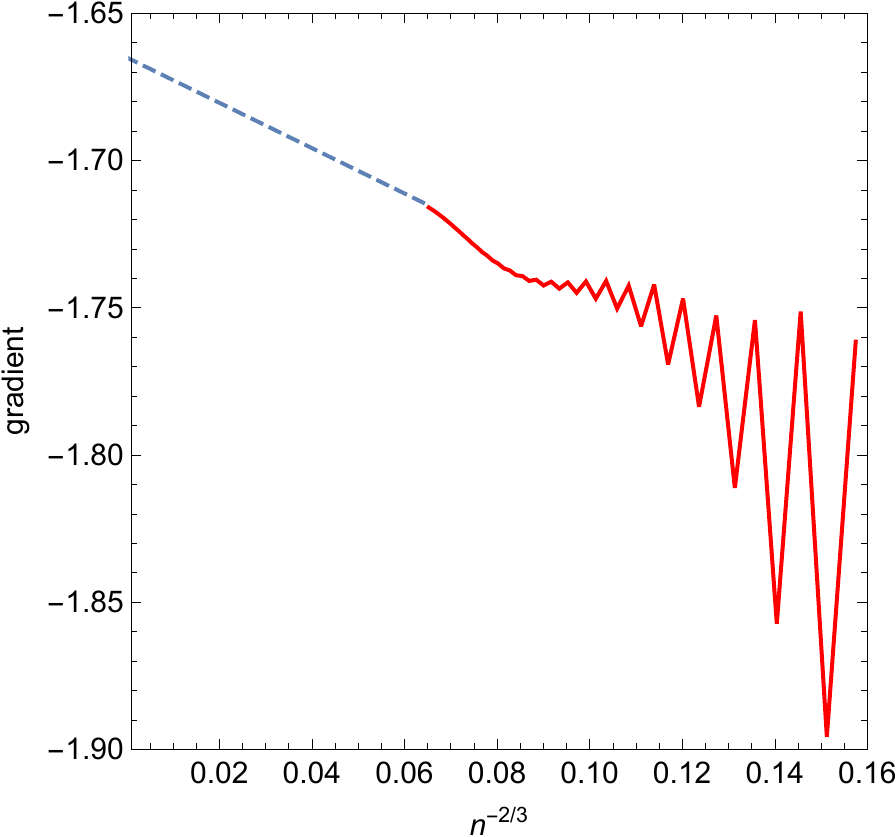} \qquad   \includegraphics[width=2.5in]{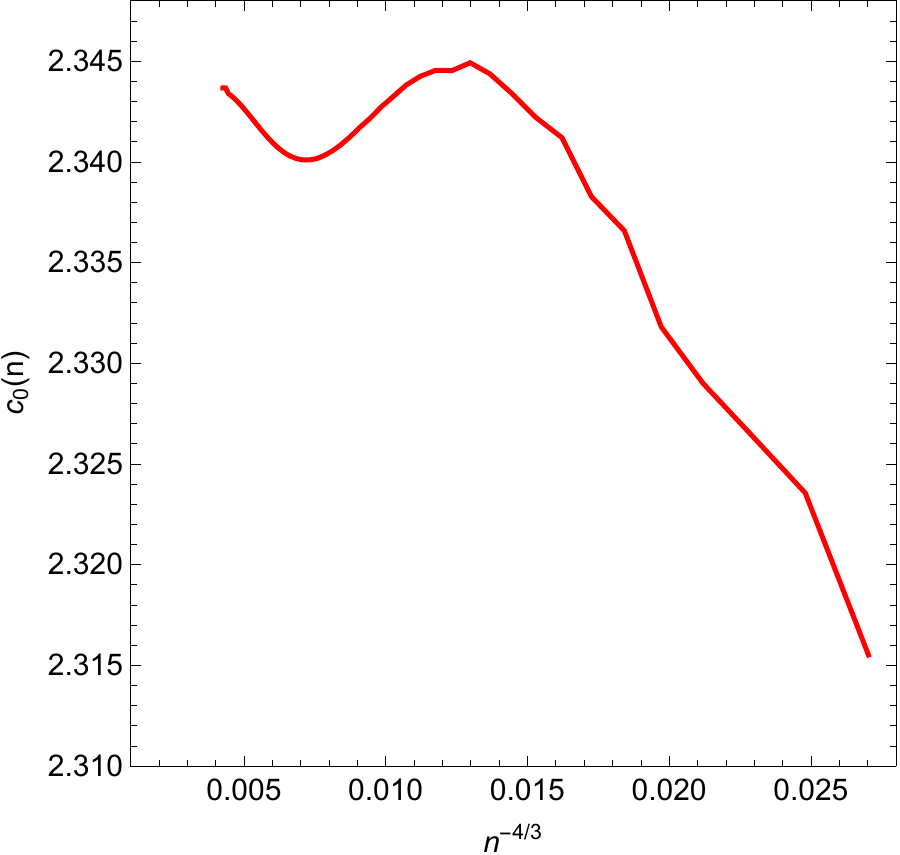} 
   \caption{(Left)Plot of local gradient against $n^{-2/3}$ and (right) plot of $c_0(n)$ against $n^{-4/3}$ for the square-ice series.}
   \label{fig:c0-si}
\end{figure}
%

\begin{figure}[h] 
    \includegraphics[width=2.5in]{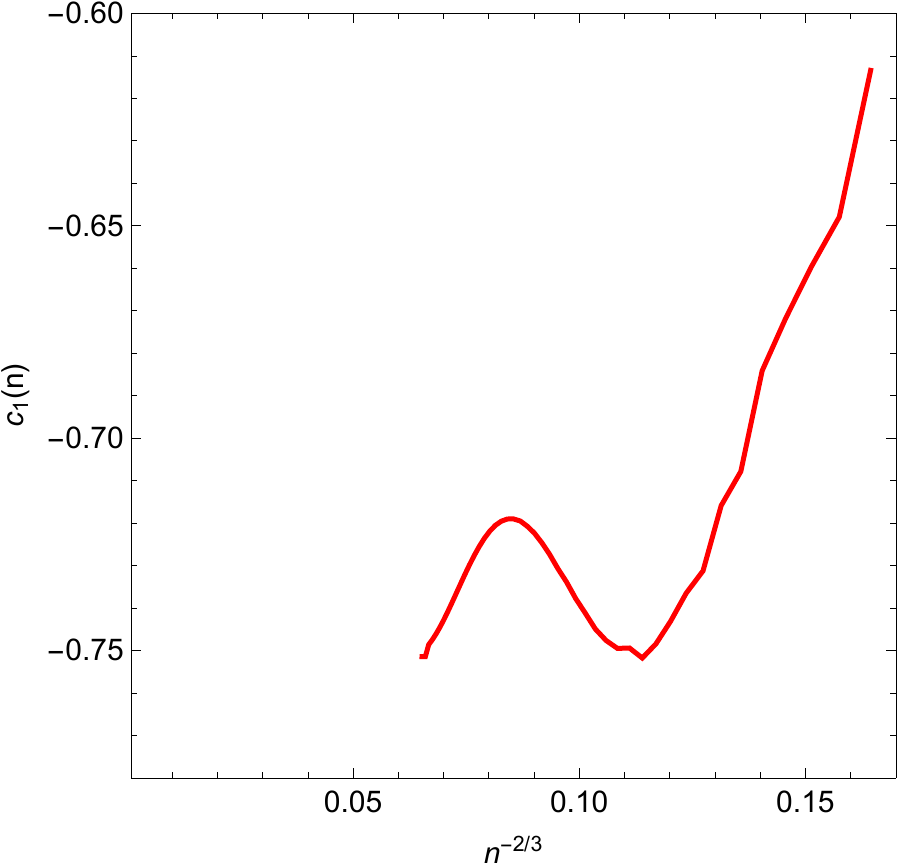} \qquad
%
     \includegraphics[width=2.5in]{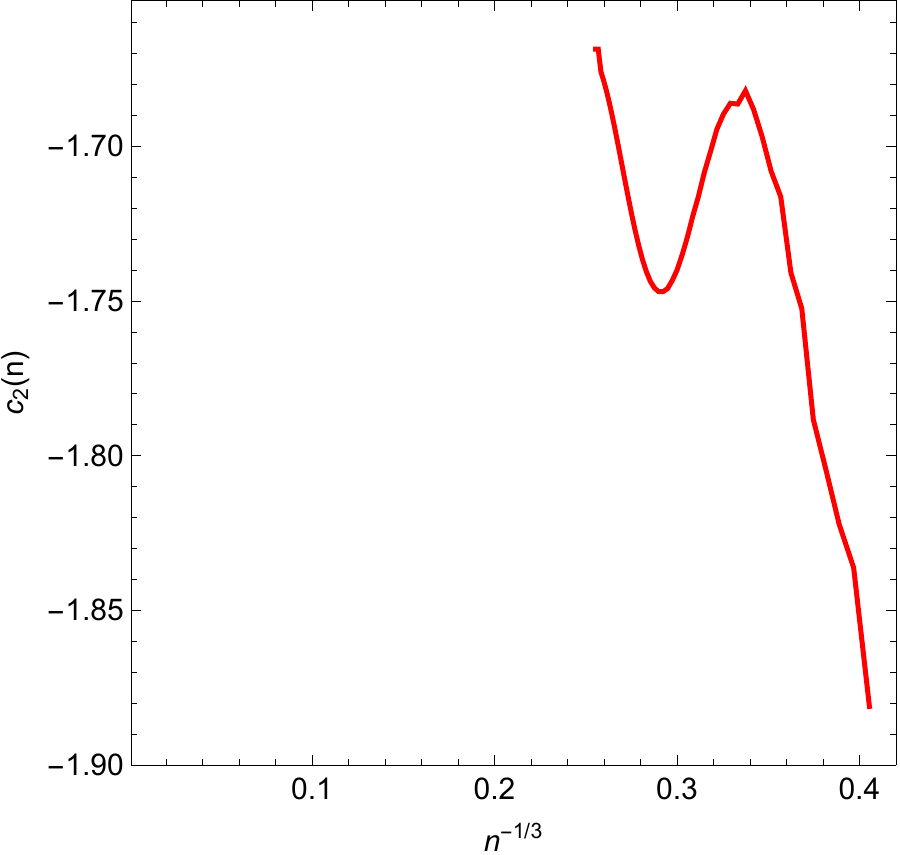} 
   \caption{Plot of $c_1(n)$  and $c_2(n)$ against $n^{-2/3}$ (resp. $n^{-1/3}$) for the square ice series.}
     \label{fig:c1-si}
\end{figure}
\subsection{Behaviour of $a_\ell(n)$ for $\ell > 1.$}
From our Monte Carlo work we concluded in Eq. \eqref{formulafinal} that  for $\ell\gg n^{1/3}$ $$a_\ell(n) \sim A_\ell\ \mu^{n^{2/3}}n^{g_\ell}$$ where $A_\ell$ and $g_\ell$ are $\ell$-dependent, while $\mu$ is not. For this investigation we can make the weaker assumption that the exponent $2/3$ can be positive exponent $\theta,$ as we will eliminate this dominant term. While our series analysis is not accurate enough to give a good estimate of $g_\ell$ directly (as shown above), we instead focus on $g_\ell-g_1.$ One has
\begin{equation} \label{gl}
\hat{a}_\ell(n) \equiv \frac{a_\ell(n)}{a_1(n)} \sim \frac{A_\ell}{A_1}n^{g_\ell-g_1},
\end{equation}
and so the exponent $\hat{g_\ell}=g_\ell-g_1$ can be estimated from the ratios of successive terms $\hat{a}_\ell(n).$ That is to say,
\begin{equation} \label{ratios}
r_\ell(n) \equiv \frac{\hat{a}_\ell(n)}{\hat{a}_{\ell}(n-1)} \sim 1+\frac{\hat{g_\ell}}{n}.
\end{equation}
 So a plot of $r_\ell(n)$ against $1/n$ should be linear, with slope $\hat{g_\ell},$ and with ordinate 1 as $n \to \infty.$
We show in Figure \ref{fig:ratios} the  ratios $r_\ell(n)$ plotted against $1/n$ for $\ell = 6, \, 5, \, 4, \, 3,\, 2$ reading from top to bottom. It can be seen that these ratio plots are behaving as expected, but with a small amount of curvature due to the effect of unknown higher-order terms in (\ref{ratios}). We attempt to accommodate these by calculating the local gradient $$\hat{g}_\ell(n) = n(r_\ell(n)-1) \sim \hat{g_\ell} + o(1).$$ In fact, it appears empirically that the term $o(1)$ can be replaced by $O(1/n),$ as plots of $\hat{g}_\ell(n)$ against $1/n$ appear to be essentially linear. In this way we estimate $$g_\ell \approx 0.0, \,\, 0.058, \,\, 0.17, \,\, 0.37, \,\, 0.64 $$ for $\ell = 2, \,\, 3, \,\, 4, \,\, 5, \,\, 6$ respectively. These differences lie somewhere between those obtained from fit 2 and fit 3 in our Monte Carlo analysis. Note that for $\ell>3$, we do not have exact numbers for $n>\ell^3$, so the above analysis can be taken seriously only for $\ell\leq 3$.  As the series analysis is independent of any assumptions except the form  \eqref{formulafinal}, we might expect series analysis to be more accurate for this parameter.

 \rowcolors{2}{gray!13}{}
\begin{table}
\[
\begin{array}{rrrr}
n & \text{sf}(n)\hspace*{1.8cm} & \text{mc}_1(n) \hspace*{1cm}& \text{\% error mc}_1(n)\\ \hline
 61 & 5.08349035674\times 10^{13} & 50834979702073 & 0.00641281 \\
 62 & 7.460434311\times 10^{13} & 74604412596394 & 0.0128265 \\
 63 & 1.092771318\times 10^{14} & 109276600121877 & 0.0188305 \\
 64 & 1.597623083\times 10^{14} & 159761033617959 & 0.0245731 \\
 65 & 2.3313927896\times 10^{14} & 233136871953374 & 0.030269 \\
 66 & 3.39600366875\times 10^{14} & 339595034177620 & 0.0362506 \\
 67 & 4.9379657155\times 10^{14} & 493785801126495 & 0.0421327 \\
 68 & 7.1674931\times 10^{14} & 716736440905024 & 0.0478466 \\
 69 & 1.0385930349\times 10^{15} & 1038570180194263 & 0.053417 \\
 70 & 1.5023341234\times 10^{15} & 1502380905370668 & 0.0590596 \\ \hline
\end{array}
\]
\caption{Comparing the series estimates, sf$_1(n)$, with the Monte Carlo estimates, mc$_1(n)$, for $a_1(n)$. The differences of the two estimates are consistently lower than the error in column 3 by an order of magnitude.} \label{seriescompare}
\end{table}

\clearpage

\begin{figure}[ht] 
 \centering
   \includegraphics[width=2.5in]{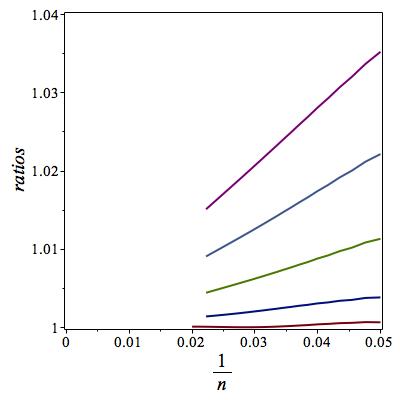} 
   \caption{Plot of ratios $r_\ell(n)$ against $1/n$ for $\ell = 6, \, 5, \, 4, \, 3,\, 2$ reading from top to bottom.}
   \label{fig:ratios}
\end{figure}

\section{Concluding Remarks}

In this paper, we have addressed several aspects of the square-ice analogue of plane partitions. Our exact enumerations have non-trivially extended the numbers provided by Young.  As expected, the asymptotic behaviour is similar to that of plane partitions. We showed this by establishing the leading asymptotic behaviour and then using Monte Carlo simulations to provide evidence for behaviour of the form given in Eq. \eqref{formulafinal}. The lack of a generating function makes it an ideal testing ground for the series extension methods that one of us (AG) has developed. In this context, our Monte Carlo simulations provide an independent check on the method. 

The exact data as well as our Monte Carlo simulations clearly indicate that for $n>1$ and $\ell>\ell'$, one has $a_\ell(n)>a_{\ell'}(n)$. However, we have not proved this statement and will leave it as an open conjecture. Conjecture  \ref{conj1} also remains open and suggests the existence of a new statistic that might enable one to prove the conjecture. Of course, it remains to be seen if one can find explicit formulae for the generating functions for $a_\ell(n)$. 

\noindent \textbf{Acknowledgments:} We thank Nicolas Destainville for useful conversations as well as sharing his Monte Carlo code for solid partitions. We are grateful to Jim Propp, Rick Kenyon, Ben Young and others members  of the domino forum for drawing our attention to this problem.
\clearpage
\appendix
  \rowcolors{2}{gray!13}{}

\begin{landscape}
\section{Numbers from exact enumeration}

\mbox{}
\vfill
  \rowcolors{2}{gray!13}{}
\hspace*{-12pt}{\scriptsize 
\begin{tabular}{c|rrrrrrrrrrrrrrrr}
$n$&0& 1 & 2 & 3 & 4 & 5 & 6& 7& 8& 9& 10 & 11 & 12 & 13 &14 &15 \\[2pt]  \hline 
$a_1(n)$\phantom{\Big|}& 1 & 1 & 4 & {\color{red} 10} &{\color{red} 24} & \color{red} 51 & \color{red} 109 & \color{red} 222 & \color{red} 452 & \color{red} 890 & \color{red} 1732 & \color{red} 3298 & \color{red}  6204 &
  \color{red}  11470 & \color{red} 20970 & \color{red} 37842 \\
 $a_2(n)$\phantom{\Big|}&1 & 2 & 5 &  12 & 29 & {\color{red} 64} & \color{red} 139 & \color{red} 286 & \color{red} 582 & \color{red} 1148 & \color{red} 2227 &\color{red}  \color{red} 4234 & \color{red} 7950 &
  \color{red}  14692 & \color{red} 26842 &\color{red}  48438 \\
$a_3(n)$\phantom{\Big|}& 1 & 3 & 7 & 19 & 44 &  98 &  213 & \color{red} 448 & \color{red} 918 & \color{red} 1832 & \color{red} 3584 & \color{red} 6882 & \color{red} 13012
   & \color{red} 24220 & \color{red} 44480 & \color{red} 80678 \\
 $a_4(n)$\phantom{\Big|}&1 & 4 & 10 & 28 & 68 & 158 & 350 & 750 & 1559 & \color{red} 3170 & \color{red} 6292 & \color{red}12252 &
   \color{red} 23445 & \color{red} 44164 & \color{red} 81995 & \color{red} 150288 \\
 $a_5(n)$\phantom{\Big|}&1 & 5 & 14 & 40 & 103 & 247 & 567 & 1252 & 2668 & 5539 & 11214 &\color{red}  22247 &
 \color{red}  43300 & \color{red} 82871 & \color{red}156152 & \color{red}290202 \\
$a_6(n)$\phantom{\Big|}& 1 & 6 & 19 & 56 & 152 & 378 & 898 & 2042 & 4476 & 9526 & 19740 & 39978 &
   79342 & \color{red} 154650 & \color{red} 296489 & \color{red} 560022 \\
$a_7(n)$\phantom{\Big|}& 1 & 7 & 25 & 77 & 219 & 567 & 1392 & 3263 & 7354 & 16048 & 34055 & 70503
   & 142842 & 283832 & 554196 & \color{red} 1065070 \\
$a_8(n)$\phantom{\Big|}& 1 & 8 & 32 & 104 & 309 & 834 & 2116 & 5114 & 11849 & 26520 & 57620 &
   121950 & 252256 & 511180 & 1016878 & 1989150 \\
$a_9(n)$\phantom{\Big|}& 1 & 9 & 40 & 138 & 428 & 1204 & 3159 & 7870 & 18747 & 43036 & 95729 &
   207125 & 437402 & 903914 & 1831938 & 3647757 \\
$a_{10}(n)$\phantom{\Big|}& 1 & 10 & 49 & 180 & 583 & 1708 & 4637 & 11906 & 29158 & 68652 & 156336 &
   345780 & 745450 & 1570920 & 3243407 & 6573672 \\
$a_{11}(n)$\phantom{\Big|}& 1 & 11 & 59 & 231 & 782 & 2384 & 6699 & 17726 & 44627 & 107763 & 251213
   & 567936 & 1249864 & 2685688 & 5648561 & 11652141 \\
$a_{12}(n)$\phantom{\Big|}& 1 & 12 & 70 & 292 & 1034 & 3278 & 9534 & 25998 & 67276 & 166602 & 397542
   & 918580 & 2063435 & 4520696 & 9684744 & 20332156 \\
$a_{13}(n)$\phantom{\Big|}& 1 & 13 & 82 & 364 & 1349 & 4445 & 13379 & 37596 & 99983 & 253894 &
   620074 & 1464231 & 3357015 & 7498084 & 16360443 & 34952692 \\
$a_{14}(n)$\phantom{\Big|}& 1 & 14 & 95 & 448 & 1738 & 5950 & 18528 & 53650 & 146605 & 381704 &
   954023 & 2302014 & 5386122 & 12263456 & 27250699 & 59239788 \\
$a_{15}(n)$\phantom{\Big|}& 1 & 15 & 109 & 545 & 2213 & 7869 & 25342 & 75605 & 212253 & 566525 &
   1448904 & 3572062 & 8528301 & 19792143 & 44785054 & 99055075 \\
$a_{16}(n)$\phantom{\Big|}& 1 & 16 & 124 & 656 & 2787 & 10290 & 34260 & 105290 & 303628 & 830660 &
   2173572 & 5474290 & 13335035 & 31540550 & 72667194 & 163510356 \\
$a_{17}(n)$\phantom{\Big|}& 1 & 17 & 140 & 782 & 3474 & 13314 & 45811 & 144998 & 429428 & 1203961 &
   3222775 & 8290859 & 20603148 & 49659502 & 116480131 & 266609605 \\
 $a_{18}(n)$\phantom{\Big|}& 1 & 18 & 157 & 924 & 4289 & 17056 & 60627 & 197578 & 600837 & 1725998 &
   4725599 & 12415980 & 31472081 & 77292052 & 184550272 & 429643458 \\
 $a_{19}(n)$\phantom{\Big|}&1 & 19 & 175 & 1083 & 5248 & 21646 & 79457 & 266540 & 832108 & 2448742 &
   6856258 & 18395107 & 47555212 & 118985326 & 289169418 & 684647303 \\
$a_{20}(n)$\phantom{\Big|}& 1 & 20 & 194 & 1260 & 6368 & 27230 & 103182 & 356174 & 1141253 & 3439858
   & 9847768 & 26976044 & 71115601 & 181255026 & 448307314 & 1079349306
\end{tabular}
}
\captionof{table}{Numbers in red are those for which the (generic) formula for $a_\ell(n)$ is anticipated to fail. The numbers have been checked for $n\leq 9$ and $\ell \leq 20$.}
\vfill
\end{landscape}

\begin{table}
{\scriptsize
\begin{tabular}{c|r|r|r|r|r|r}
$n$&$a_1(n)$\hspace*{0.5cm}&$a_2(n)$\hspace*{0.5cm}&$a_3(n)$\hspace*{0.5cm}&$a_4(n)$\hspace*{0.5cm}&$a_5(n)$\hspace*{0.5cm}&$a_6(n)$\hspace*{0.5cm} \\[2pt] \hline
 1 & 1 & 2 & 3 & 4 & 5 & \phantom{\Big |}6 \\
 2 & 4 & 5 & 7 & 10 & 14 & 19 \\
 3 & 10 & 12 & 19 & 28 & 40 & 56 \\
 4 & 24 & 29 & 44 & 68 & 103 & 152 \\
 5 & 51 & 64 & 98 & 158 & 247 & 378 \\
 6 & 109 & 139 & 213 & 350 & 567 & 898 \\
 7 & 222 & 286 & 448 & 750 & 1252 & 2042 \\
 8 & 452 & 582 & 918 & 1559 & 2668 & 4476 \\
 9 & 890 & 1148 & 1832 & 3170 & 5539 & 9526 \\
 10 & 1732 & 2227 & 3584 & 6292 & 11214 & 19740 \\
 11 & 3298 & 4234 & 6882 & 12252 & 22247 & 39978 \\
 12 & 6204 & 7950 & 13012 & 23445 & 43300 & 79342 \\
 13 & 11470 & 14692 & 24220 & 44164 & 82871 & 154650 \\
 14 & 20970 & 26842 & 44480 & 81995 & 156152 & 296489 \\
 15 & 37842 & 48438 & 80678 & 150288 & 290202 & 560022 \\
 16 & 67572 & 86509 & 144697 & 272150 & 532430 & 1043404 \\
 17 & 119368 & 152902 & 256775 & 487388 & 965395 & 1919708 \\
 18 & 208943 & 267783 & 451305 & 863887 & 1731351 & 3491081 \\
 19 & 362389 & 464766 & 786008 & 1516592 & 3073660 & 6280514 \\
 20 & 623438 & 800095 & 1357414 & 2638648 & 5404984 & 11185375 \\
 21 & 1064061 & 1366512 & 2325540 & 4552488 & 9420512 & 19734004 \\
 22 & 1802976 & 2316840 & 3954366 & 7792566 & 16282463 & 34509347 \\
 23 & 3033711 & 3900502 & 6676369 & 13239698 & 27922063 & 59847208 \\
 24 & 5071418 & 6523432 & 11196599 & 22336630 & 47527430 & 102976946 \\
 25 & 8424788 & 10841282 & 18657454 & 37433466 & 80331385 & 175877782 \\
 26 & 13913192 & 17909533 & 30901434 & 62337628 & 134873275 & 298279841
   \\
 27 & 22847028 & 29416966 & 50884452 & 103186612 & 225015223 & 502496682
   \\
 28 & 37315678 & 48055443 & 83327163 & 169824540 & 373141724 & 841161007
   \\
 29 & 60631940 & 78093926 & 135733071 & 277967860 & 615224276 &
   1399559416 \\
 30 & 98030644 & 126276743 & 219978688 & 452594316 & 1008792896 &
   2315201903 \\
 31 & 157743554 & 203211038 & 354780782 & 733229626 & 1645443771 &
   3808746574 \\
 32 & 252671288 & 325518314 & 569519349 & 1182159039 & 2670372299 &
   6232651705 \\
 33 & 402944731 & 519138982 & 910130189 & 1897140990 & 4312780664 &
   10147431024 \\
 34 & 639871871 & 824414851 & 1448166991 & 3031012912 & 6933014899 &
   16440685315 \\
 35 & 1011956958 & 1303853212 & 2294680459 & 4821835750 & 11095408859 &
   26512248644 \\
 36 & 1594100512 & 2053981256 & 3621419828 & 7639072393 & 17680429741 &
   42561099330 \\
 37 & 2501559132 & 3223352798 & 5693103210 & 12054120068 & 28056800955 &
   68028465562 \\
 38 & 3911136893 & 5039865872 & 8916408778 & 18947689292 & 44344779210 &
   108279807765 \\
 39 & 6093172867 & 7852029282 & 13914109052 & 29672809254 & 69817667843
   & 171651101620 \\
 40 & 9459795828 & 12191192807 & 21636960372 & 46301523560 &
   109512215347 & 271048865628 \\
 41 & 14637397882 & 18865058704 & 33532084406 & 71997231090 &
   171153951432 & 426389614752 \\
 42 & 22575337525 & 29097916032 & 51795716561 & 111575067538 &
   266555833407 & 668307945618 \\
 43 & 34708392976 & 44740293582 & 79751566012 & 172343093538 &
   413726582860 & 1043776858764 \\
 44 & 53199143209 & 68581738911 & 122415827920 & 265361653081 &
   640040090348 & 1624602354318 \\
 45 & 81298470388 & 104816149708 & 187338790559 & 407324082526 &
   986987600985 & 2520227376672 \\
   46 & 123880767618 & 159732599729 \\
 47 & 188236334008 & 242738329372 \\
 48 & 285242287944 & 367870426468 \\
 49 & 431088527694 & 556024400588 \\
 50 & 649816920320 & 838232884647 \\
  51 & 977048352353 \\
 52 & 1465442861255 \\
 53 & 2192681711158 \\
 54 & 3273114322046 \\
 55 & 4874718706124 \\
 56 & 7243754365560 \\
 57 & 10740528588174 \\
58 & 15891194045343 \\
59 & 23462627747108 \\
60 & 34570490892429 \\[3pt] \hline
\end{tabular}
}
\caption{Results from Exact Enumeration}\label{exacttable}
\end{table}
\clearpage

\section{A class of restricted plane partitions}

A plane partition is an array of non-negative integers $h_{i,j}$ that are weakly decreasing along both rows and columns i.e.,
\[
h_{i+i,j} \leq h_{i,j}\quad \text{and} \quad h_{i,j+1} \leq h_{i,j}\quad \text{for all } i,j\geq 1\ .
\]
The volume of a plane partition is defined to be the sum of all entries in the array i.,e. $\sum_{i,j} h_{i,j}$ and let $p_2(n)$ denote the number of plane partitions with volume $n$. 
Let $pr(n)$ denote the subset of plane partitions where  one imposes the stronger condition
\[
h_{i+i,j} = h_{i,j}- e \quad \text{and} \quad h_{i,j+1} = h_{i,j}-e\quad \text{for all } i,j\geq 1\ ,
\]
where $e=0$ or $e=1$. 

\noindent The first few numbers are 
\begin{center}
\begin{tabular}{c|r|r|r|r|r|r|r|r  |r | r}
$n$& 1 & 2 & 3 & 4 & 5 & 6 & 7 & 8 & 9~ & 10~\\
$pr(n)$ &  1&2&3&6&10&18&30&41&63 & 102 \\
$p_2(n)$ & 1&3& 6& 13& 24& 48& 86& 160& 282 & 500
\end{tabular}
\end{center}
It is easy to see that 
\begin{equation}
pr(n) <  p_2(n) \text{ for } n>1 \ .
\end{equation}
We also have for $n\gg 1$ that
\[
\log a_1(n) < 4 \log pr(n/4).\ 
\]
We thus have
\begin{equation}
\log pr\left(\tfrac{n}4\right) > \frac{c_0}4 \ n^{2/3}\quad \text{or} \quad \log pr(n) > \frac{c_0}{4^{1/3}} \ n^{2/3}.
\end{equation}
We thus obtain the asymptotic bound as $n\rightarrow\infty$:
\begin{equation}
\frac1{4^{1/3}} n^{-2/3}\log a_1(n)  < n^{-2/3} \log pr(n) < n^{-2/3} \log p_2(n)\ ,
\end{equation}
or equivalently 
\begin{equation}
\boxed{
\frac{c_0}{4^{1/3}}  < n^{-2/3} \log pr(n) <  \tfrac32 (2\zeta(3))^{1/3}\ .
}
\end{equation}
Our estimate of $c_0\approx 2.344$ thus enables us to set a lower bound for the asymptotic behaviour of $pr(n)$.

\end{document}